\documentclass[preprint,showkeys,superscriptaddress,nofootinbib,amsmath,amssymb,aps,pra,longbibliography]{revtex4-1}
\usepackage[utf8x]{inputenc}
\usepackage{amsmath,amsthm}
\usepackage{dcolumn}
\usepackage{bm}
\usepackage[colorlinks=true,citecolor=blue,urlcolor=blue]{hyperref}
\usepackage{graphicx}
\usepackage{subcaption}
\usepackage{braket}
\usepackage{color}
\usepackage[T1]{fontenc}
\usepackage[usenames,dvipsnames,svgnames,table]{xcolor}

\newcommand{\be}{\begin{equation}}
\newcommand{\ee}{\end{equation}}
\newcommand{\ba}{\begin{eqnarray}}
\newcommand{\ea}{\end{eqnarray}}

\newcommand{\tr}{\operatorname{Tr}}

\newtheorem{proposition}{Proposition}

\newtheorem{Lemma}{Lemma}

\begin{document}
 \newwrite\bibnotes
    \def\bibnotesext{Notes.bib}
    \immediate\openout\bibnotes=\jobname\bibnotesext
    \immediate\write\bibnotes{@CONTROL{REVTEX41Control}}
    \immediate\write\bibnotes{@CONTROL{%
    apsrev41Control,author="08",editor="1",pages="1",title="0",year="1"}}
     \if@filesw
     \immediate\write\@auxout{\string\citation{apsrev41Control}}%
    \fi

\title{"All-versus-nothing" proof of genuine tripartite steering and 
    entanglement certification in the two-sided device-independent
    scenario}

\author{Shashank Gupta}
\email{shashankg687@bsoe.res.in}
\affiliation{S.   N.  Bose  National Centre  for Basic  Sciences, Salt
 Lake,    Kolkata    700    106,   India}  

 \author{Debarshi      Das}
   \email{dasdebarshi90@gmail.com}
\affiliation{S.   N.  Bose  National Centre  for Basic  Sciences, Salt
 Lake,    Kolkata    700    106,   India} 
  
\author{C. Jebarathinam} 
 \email{jebarathinam@gmail.com}
 \affiliation{Center for Theoretical Physics, Polish Academy of Sciences, Aleja Lotników 32/46, 02-668 Warsaw, Poland}

 \author{Arup Roy}  
 \email{arup145.roy@gmail.com}
\affiliation{S.   N.  Bose  National Centre  for Basic  Sciences, Salt
 Lake,    Kolkata    700    106,   India}

\author{Shounak Datta}
\email{shounak.datta@bose.res.in}
\affiliation{S.   N.  Bose  National Centre  for Basic  Sciences, Salt
  Lake,    Kolkata    700    106,   India}   

\author{A. S. Majumdar}
\email{archan@bose.res.in}
 \affiliation{S. N. Bose National Centre for
  Basic Sciences, Salt Lake, Kolkata 700 106, India} 
  
\date{\today}

\begin{abstract}

We consider the task of certification of genuine entanglement of tripartite states. For this purpose, we first present an  ``all-versus-nothing'' proof of  genuine tripartite  Einstein-Podolsky-Rosen (EPR) steering by demonstrating the non-existence of a hybrid local hidden state (LHS) model in the tripartite network as a motivation to our main result.  A full logical contradiction of the predictions of the hybrid LHS model  with quantum mechanical outcome statistics for any three-qubit generalized Greenberger-Horne-Zeilinger (GGHZ) states and pure W-class states is shown. Using logical contradiction, we can distinguish between the GGHZ and W-class state in two-sided device-independent (2SDI) steering scenario. We next formulate a 2SDI steering inequality which is a generalisation of the fine-grained steering inequality (FGI) derived in \cite{PKM14} for  the tripartite scenario. We show that the maximum quantum violation of this tripartite FGI  can be used to certify genuine entanglement of three-qubit pure states. 

\end{abstract}

\keywords{Quantum entanglement, Entanglement certification, Quantum steering, Quantum tomography}

\maketitle

\section{Introduction}

The presence of entanglement among spatially separated parties is one of the most intriguing  phenomena in quantum physics. A bipartite quantum state is called entangled  if and only if (iff) it is not separable. A pure multipartite quantum state is called genuinely multipartite entangled \cite{GUHNE91} iff it  is not separable with respect to any bipartition. Studies on multipartite entanglement have gained a lot of attention due to their foundational significance as well as their information theoretic applications, for example, in extreme spin squeezing \cite{Soren01}, high sensitivity metrology tasks \cite{Hyllus12,Toth12}, quantum computing using cluster states \cite{Raussendorf1}, measurement-based quantum computation \cite{Briegel09} and multiparty quantum network \cite{Murao99,Hillery99}. 
 
 Due to its complex structure, the detection of genuine multipartite entanglement is quite difficult to realize experimentally. Apart from the usual state tomography, various witnesses \cite{ALS1,Bourennane4,Bru2002,GUHNE91,Huber14,Clivaz2017,Li2017} have been proposed for detecting genuine multipartite entanglement (GME). However, these processes have experimental limitations due to the requirement of precise control over the systems. To overcome these limitations, one may use device-independent (DI)  witnesses \cite{Bowles18} of GME  which are based on observing measurement statistics without any characterization of the experimental devices. Several  multipartite Bell-type inequalities have been proposed to certify GME  in the  DI scenario  \cite{Seevinck01,Nagata02,BGL+11,BPR+17,Zwerger19,AST+19,Ananda20} which requires no trust/characterization of any observer's devices. 
Naturally, the question is raised whether one can certify the presence of genuine multipartite entanglement in a scenario where some of the observers' devices are not characterized, i.e., when the scenario is partially device-independent. 

The above-mentioned partially device-independent scenario is related to the concept of EPR steering which was initially proposed by Schrödinger \cite{Schrdinger35}. Later, in 2007, Wiseman \textit{et al.} \cite{Wiseman7} formulated this phenomenon in terms of an operational task. In the bipartite scenario, EPR steering (or, quantum steering) implies the possibility of producing a different sets of states at one party's (say, Bob) end by performing local quantum measurements of any two non-commuting observables on the other spatially separated party's (say, Alice) end. ERP steering is defined as the non-existence of local hidden state (LHS) models to describe the conditional states produced at Bob's side. Here, the outcome statistics of one subsystem (which is being `steered', in the present case Bob) are produced from the specific local measurements on other subsystem. Quantum steering is useful for semi-device-independent certification of all pure bipartite maximally entangled states \cite{Shrotriya21}, indefinite causal order \cite{Bavaresco2019}  and incompatible measurements \cite{Quintino14,Uola14,sarkar2021}. In partially device-independent scenario comprising many parties, genuine steering (nonexistence of hybrid-local-hidden-state model) demonstrates the presence of genuine multipartite entanglement \cite{CSA+15,Gupta21,Guptas21}. Now the question is: how to detect genuine steering in such scenarios?

In quantum theory, several no-go theorems depict the failure of certain physical models
to describe the outcome statistics produced in quantum mechanics. For instance, the well-known
Greenberger-Horne-Zeilinger (GHZ) theorem \cite{GHZ,Greenberger1990} rules out the existence of a local hidden variable (LHV) model for the GHZ state by presenting a full contradiction of the predictions of the LHV model with  the outcome statistics produced in quantum mechanics. The GHZ theorem has also been experimentally verified \cite{Deng10,Reid88}. Similarly, non-locality without inequalities was demonstrated by taking into account the contradiction of LHV 'elements of reality' with quantum mechanics \cite{Hardy92,Hardy93,Goldstein94}. Such ``all-versus-nothing'' proof of Bell-nonlocality rules out the possibility of local-hidden variable (LHV) models more uncompromisingly than Bell inequalities.

Recently,  EPR steering of any arbitrary two-qubit pure entangled state has been demonstrated without invoking any steering inequality \cite{CSX+16}. In the present work our first motivation is to extend the GHZ theorem for any three-qubit  generalized Greenberger-Horne-Zeilinger (GGHZ) state and any pure W-class state to the genuine tripartite steering scenario. Genuine tripartite steering is demonstrated through the violation of  hybrid LHS models \cite{CSA+15}, whereas, steering in general, is demonstrated through the violation of LHS models. In the genuine steering case, the models are hybrid in the sense that the hidden variable may not predetermine the state of one of the subsystems. In other words, the local state of one of the subsystems is steerable, whereas, in case of  LHS models, the hidden variable predetermines the state of each subsystem, and there is no scope of steering even at the level of a local subsystem. This forms the first step towards certifying genuine tripartite entanglement in the semi-DI scenario through sharp logical contradiction.

With the above goal, here we present an ``all-versus-nothing'' proof of genuine tripartite EPR steering for any three-qubit pure generalized Greenberger-Horne-Zeilinger (GGHZ) state and pure W-class state in both one-sided and two-sided device-independent scenario, where two or one of the three parties perform characterized measurements and the remaining parties perform black-box measurements. Next, to capture the genuine quantum steering correlation present in GGHZ states, we propose a 2SDI tripartite steering inequality which is a generalization of the fine-grained steering inequality in the bipartite scenario \cite{PKM14}. Fine-grained steering inequalities are derived using the fine-grained uncertainty relation \cite{OW10} and provide tighter steering criteria \cite{Chowdhury14,Chowdhury15} over steering inequalities based on Heisenberg's uncertainty relation  \cite{Reid89,Saunders2010}, as well as over steering inequalities based on entropic uncertainty relations \cite{Walborn11}. Quantum violation of our proposed 2SDI tripartite fine-grained steering inequality  implies tripartite steering in the 2SDI scenario. We then demonstrate that the maximum quantum violation of this inequality certifies the presence of genuine entanglement of pure three-qubits in the 2SDI scenario. 
 
The paper is organized as follows. In Sec. \ref{section1}, we review the concept of genuine tripartite steering in the 2SDI scenario. In Sec. \ref{section2}, we demonstrate the non-existence of the hybrid LHS model for any three-qubit pure GGHZ state using logical contradiction. In Sec. \ref{section3}, we derive a fine-grained steering inequality for the tripartite steering case in the 2SDI scenario. The main result of the paper, {\it viz.}, certification of genuine entanglement of three-qubit pure states is shown in Sec. \ref{section4}. 
 Concluding remarks are presented in Sec. \ref{section6}.

\section{Tripartite Quantum Steering in the two-sided \\
device-independent scenario}\label{section1}

We begin by recapitulating the definitions of tripartite and genuine tripartite steering \cite{Cavalcanti11,Cavalcanti16,Bihalan18} . In this context two scenarios arise: 1) one-sided device-independent (1SDI) scenario and 2) two-sided device-independent (2SDI) scenario. For the  purpose of the present study, we will restrict ourselves to the 2SDI scenario.

Consider a tripartite steering scenario where three spatially separated parties, say Alice, Bob and Charlie, share an unknown quantum system $\rho_{A'B'C} \in \mathcal{B}(\mathcal{H}_{A'} \otimes \mathcal{H}_{B'} \otimes \mathcal{H}_C)$ with the local Hilbert space dimension of Alice's subsystem and that of Bob's subsystem being arbitrary and the local Hilbert space dimension of Charlie's subsystem being fixed (X' represents uncharacterized subsystem, $X' \in \{A',B'\}$). Here, $\mathcal{B}(\mathcal{H}_{A'} \otimes \mathcal{H}_{B'} \otimes \mathcal{H}_C)$ stands for the set of all density operators acting on the Hilbert space $\mathcal{H}_{A'} \otimes \mathcal{H}_{B'} \otimes \mathcal{H}_C$. Alice performs a set of positive operator-valued measurements (POVM) $X_x$ with outcomes $a$. Here $x$ $\in$ $\{0,1,2, \cdots, n_A\}$ denotes the measurement choices of Alice and $a$ $\in$ $\{0, 1, \cdots, d_A\}$. On the other hand, Bob performs a set of positive operator-valued measurements (POVM) $Y_y$  with outcomes $b$. Here $y$ $\in$ $\{0,1,2, \cdots, n_B\}$ denotes the measurement choices of Bob and $b$ $\in$ $\{0, 1, \cdots, d_B\}$.
These local measurements by Alice and Bob prepare the set of conditional states on Charlie's side.  

The above scenario is called 2SDI since the POVM elements $\{M^{A'}_{a|X_x}\}_{a,X_x}$ (where $M^{A'}_{a|X_x} \geq 0$ $\forall a, x$; and $\sum_{a} M^{A'}_{a|X_x} = \mathbb{I}$ $\forall x$) associated with Alice's measurements and the POVM elements $\{M^{B'}_{b|Y_{y}}\}_{b,Y_y}$ (where $M^{B'}_{b|Y_{y}} \geq 0$ $\forall b, y$; and $\sum_{b} M^{B'}_{b|Y_{y}} = \mathbb{I}$ $\forall y$) associated with Bob's measurements are unknown. The steering scenario is characterized by the assemblage $\{\sigma^{C}_{a,b|X_x,Y_y}\}_{a,X_x,b,Y_y}$ which is the set of unnormalized conditional states on Charlie's side. Each element in the assemblage  is given by $\sigma^{C}_{a,b|X_x,Y_y}=P(a,b|X_x,Y_y)\varrho^{C}_{a,b|X_x,Y_y}$,  where $P(a,b|X_x,Y_y)$ is the conditional probability of getting the outcome $a$ and $b$ when Alice performs the measurement $X_x$ and Bob performs measurement $Y_y$ respectively; $\varrho^{C}_{a,b|X_x,Y_y}$ is the normalized conditional state on Charlie's end. Quantum theory predicts that all valid assemblages should satisfy the following criterion:
\begin{align}
\sigma^{C}_{a,b|X_x,Y_y}&= \tr_{A'B'}\Big[ \big( M^{A'}_{a|X_x} \otimes M^{B'}_{b|Y_y} \otimes \openone \big) \rho_{A'B'C} \Big] \nonumber \\ &\forall \hspace{0.2cm} \sigma^{C}_{a,b|X_x,Y_y} \in \{\sigma^{C}_{a,b|X_x,Y_y}\}_{a,X_x,b,Y_y}.
\label{assem2SDI}
\end{align}
The tripartite state imposes constraints on the observed assemblage. For example, if the shared state has no entanglement, or in other words, it is separable, then the assemblage has a local-hidden-state (LHS) model, i.e., if for all $a$, $x$, $b$, $y$, there is a decomposition of $\sigma^{C}_{a,b|X_x,Y_y}$ in the form \cite{Cavalcanti11,Cavalcanti16,Bihalan18},
\begin{equation}
\sigma^{C}_{a,b|X_x,Y_y}=\sum_\lambda P(\lambda) \, P(a|X_x,\lambda) \, P(b|Y_y,\lambda) \, \rho^{C}_\lambda ,
\label{lhs}
\end{equation}
where $\lambda$ denotes local hidden variable (LHV) which occurs with probability 
$P(\lambda) > 0$; $\sum_{\lambda} P(\lambda) = 1$; the quantum states $\rho^{C}_\lambda$
are called local hidden states (LHS) which satisfy $\rho^{C}_\lambda\ge 0$ and
$\tr\rho^C_\lambda=1$.

Now, suppose that Charlie performs a set of characterized POVMs  $Z_z$ with outcomes $c$ having the POVM elements $\{M^C_{c|Z_z}\}_{c,Z_z}$ ($M^C_{c|Z_z} \geq 0$ $\forall c, z$; and $\sum_{c} M^C_{c|Z_z} = \mathbb{I}$ $\forall z$) on $\{\sigma^{C}_{a,b|X_x,Y_y}\}_{a,X_x,b,Y_y}$. Here $z$ $\in$ $\{0,1,2, \cdots, n_C \}$ denotes the measurement choices of Charlie and $c$ $\in$ $\{0, 1, \cdots, d_C\}$. These measurements by Charlie on $\{\sigma^{C}_{a,b|X_x,Y_y}\}_{a,X_x,b,Y_y}$ produces measurement correlations $\{P(a,b,c|X_x,Y_y,Z_z)\}_{a,X_x,b,Y_y,c,Z_z}$, where 
\begin{equation}
P(a,b,c|X_x,Y_y,Z_z)= \tr \Big[  M^C_{c|Z_z}\,  \sigma^{C}_{a,b|X_x,Y_y}\Big].
\end{equation} 
The correlation $\{P(a,b,c|X_x,Y_y,Z_z)\}_{a,X_x,b,Y_y,c,Z_z}$ detects tripartite steerability iff it does not have the following LHV-LHV-LHS decomposition: 
\begin{align}
&P(a,b,c|X_x,Y_y,Z_z)= \nonumber\\
&\sum_{\lambda} P(\lambda) \, P(a|X_x,\lambda) \, P(b|Y_y, \lambda) \, P(c|Z_z, \rho^C_\lambda) 
 \forall a,x,b,y,c,z; 
\label{LHV-LHS1SDI}
\end{align}
where $P(c|Z_z, \rho^C_{\lambda})$ denotes the quantum probability of obtaining the outcome $c$ when measurement $Z_z$ is performed on LHS $\rho^C_{\lambda}$. On the other hand, if the shared state contains no genuine entanglement, or it is in the bi-separable form, as the following \cite{CSA+15}
\begin{align}
\rho_{A'B'C} &= \sum_\lambda p_\lambda^{A':B'C}\rho_\lambda^{A'}\otimes\rho_{\lambda}^{B'C} + \sum_\mu p_\mu^{B':A'C}\rho_\mu^{B'}\otimes\rho_{\mu}^{A'C} \nonumber\\
&+ \sum_\nu p_\nu^{A'B':C}\rho_\nu^{A'B'}\otimes\rho_{\nu}^{C}, 
\label{bisepstate}
\end{align}
where $p_\lambda^{A':B'C}, p_\mu^{B':A'C}$ and $p_\nu^{A'B':C}$ are probability distributions,  the assemblage (\ref{assem2SDI}) has the  form \cite{CSA+15},
\begin{align}
	\sigma^{C}_{a,b|X_x,Y_y} &= \sum_\lambda p_\lambda^{A':B'C} p(a|X_x,\lambda)\sigma_{b|Y_y,\lambda}^C \nonumber \\
	&+ \sum_\mu p_\mu^{B':A'C}p(b|Y_y,\mu)\sigma_{a|X_x,\mu}^C \nonumber \\
	&+ \sum_\nu p_\nu^{A'B':C} p(a,b|X_x,Y_y,\nu)\rho_\nu^C.
\label{assembisep}
\end{align}
The assemblage (\ref{assembisep}) contains three terms and is quite different from the assemblage (\ref{lhs}). The first term is an unsteerable assemblage from Alice to Charlie, but not necessarily from Bob to Charlie. The Charlie's assemblage is dependent on Bob's input, output and the common variable $\lambda$. The second term is unsteerable from Bob to Charlie but not necessarily from Alice to Charlie. In this case, Charlie's assemblage is dependent on Alice's input, output and the common hidden variable $\mu$. The third term has two features: (i) it is unsteerable from Alice-Bob to Charlie, (ii) the probability distribution $p(a,b|X_x,Y_y,\nu)$ arises due to local measurements performed on a possibly entangled state, and it may contain nonlocal quantum correlations. Here, the dependence of Charlie's assemblage on Alice and Bob's input and output comes only from the hidden variable $\nu$. When each element of an assemblage (\ref{assem2SDI}) can be written in the form (\ref{assembisep}), then the assemblage does not demonstrate genuine EPR steering in the 2SDI scenario but it may demonstrate steering. On the other hand,  if the assemblage (\ref{assem2SDI}) can be written in the form (\ref{lhs}), then the assemblage is unsteerable and no signature of steering is present in the 2SDI scenario. The form (\ref{assembisep}) can be viewed in terms of a hybrid LHS model.

\section{Motivation: Genuine Tripartite steering of GGHZ states}\label{section2}

The GHZ theorem that leads to "$1 = -1$", shows tripartite Bell nonlocality of the GHZ state $\ket{\psi_{GHZ}}=\frac{1}{\sqrt{2}}\left(\ket{000}+\ket{111}\right)$ in the simplest way. 
Motivated by this simple nonlocality argument, in Ref. \cite{CSX+16}, a simple demonstration of EPR steering was presented for any bipartite pure entangled state, where the LHS models lead to the logical contradiction "$2 = 1$". Here we will demonstrate that the existence of hybrid LHS models (\ref{assem2SDI}) leads to the contradiction "$2 = 1$"  in the 2SDI scenario for  any pure state that belongs to the generalized GHZ (Greenberger-Horne-Zeilinger) class having the form,
\begin{equation}
\big|\psi(\theta)_{\text{GGHZ}} \big\rangle=\cos\theta \ket{000}+ \sin\theta \ket{111}, \hspace{0.3cm} 0 < \theta <\frac{\pi}{2}
\label{GGHZ}
\end{equation}

Let us consider that Alice and Bob prepare the generalized GHZ (GGHZ) state given by Eq.(\ref{GGHZ}). They keep two particles at their possession and send the third particle to Charlie. Next,  Alice  performs  her choice of either one of two projective measurements of the observables $X_x$ (where $X_0$ = $\vec{\sigma} \cdot \hat{n}_0^A$, $X_1$ = $\vec{\sigma} \cdot \hat{n}_1^A$) and communicates the outcome $a$ $\in$ $\{0, 1\}$. Similarly,  Bob  performs his choice of either one of two projective measurements of the observables $Y_y$ (where $Y_0$ = $\vec{\sigma} \cdot \hat{n}_0^B$, $Y_1$ = $\vec{\sigma} \cdot \hat{n}_1^B$)  and  communicates the outcome $b$ $\in$ $\{0, 1\}$. Here $\vec{\sigma}$ = $(\sigma_x, \sigma_y, \sigma_z)$; $\hat{n}_0^A$, $\hat{n}_1^A$, $\hat{n}_0^B$, $\hat{n}_1^B$ are unit vectors; $\hat{n}_0^A$ $\neq$ $\hat{n}_1^A$; $\hat{n}_0^B$ $\neq$ $\hat{n}_1^B$. Henceforth, we shall denote $\vec{\sigma} \cdot \hat{n}$ by $\sigma_n$ for any unit vector $\hat{n}$.

After Alice's and Bob's measurements, a total of sixteen possible unnormalized conditional states $\sigma^{C}_{a,b|X_x,Y_y}$ (with $a$, $b$, $x$, $y$ $\in$ $\{0,1\}$) are prepared at Charlie's end. If Charlie's conditional states have hybrid LHS description,  then each of these unnormalized unconditional state can be written in the form of Eq.(\ref{assembisep}). Since the normalized conditional states at Charlie's end are pure, the assemblage is not the convex combination of the three terms of Eq. (\ref{assembisep}), but any one of the terms of Eq. (\ref{assembisep}). We find that the dependence of Charlie's assemblage on Alice and Bob's input and output may come from the common hidden variable and it is not the case that either Alice or Bob alone can change Charlie's state through their input and output choices. This is the feature of the third term of the Eq. (\ref{assembisep}). Hence, the above ensemble should satisfy the following relation:
\begin{align}
\sum_{\nu} p(\nu) \, \rho^{C}_{\nu} &= \rho^C_{\text{GGHZ}} \nonumber \\
&= \tr_{AB}\Big[\big|\psi(\theta)_{\text{GGHZ}} \big\rangle \big\langle \psi(\theta)_{\text{GGHZ}} \big| \Big]
\label{marchar}
\end{align}

Now, consider that  Alice  performs projective measurements of the following two observables: $X_0$ = $\sigma_x$, $X_1$ = $\sigma_y$. On the other hand,  Bob  performs projective measurements of the following two observables: $Y_0$ = $\dfrac{\sigma_x+\sigma_z}{\sqrt{2}}$, $Y_1$ = $\dfrac{\sigma_y+\sigma_z}{\sqrt{2}}$. In this case each of the normalized conditional states $\{\varrho^{C}_{a,b|X_x,Y_y}\}_{a,X_x,b,Y_y}$ (where $\sigma^{C}_{a,b|X_x,Y_y}=P(a,b|X_x,Y_y)\varrho^{C}_{a,b|X_x,Y_y}$) produced at Charlie's end is a pure state. For any fixed $x$ and fixed $y$, the four normalized conditional states $\varrho^{C}_{0,0|X_x,Y_y}$,  $\varrho^{C}_{0,1|X_x,Y_y}$, $\varrho^{C}_{1,0|X_x,Y_y}$, $\varrho^{C}_{1,1|X_x,Y_y}$ are four different pure states. Moreover, the normalized conditional states $\{\varrho^{C}_{a,b|X_x,Y_y}\}_{a,X_x,b,Y_y}$ satisfy the following,
\begin{align}
&\{\varrho^{C}_{0,0|X_0,Y_0}, \varrho^{C}_{0,1|X_0,Y_0}, \varrho^{C}_{1,0|X_0,Y_0}, \varrho^{C}_{1,1|X_0,Y_0} \}  \nonumber \\
&\neq \{\varrho^{C}_{0,0|X_0,Y_1}, \varrho^{C}_{0,1|X_0,Y_1}, \varrho^{C}_{1,0|X_0,Y_1}, \varrho^{C}_{1,1|X_0,Y_1} \}.
\label{con1}
\end{align}
which means no element of the set on LHS is equal to any element of the set on RHS of (\ref{con1}).
On the other hand, we obtain 
\begin{align}
& \sigma^{C}_{0,0|X_1,Y_0}=\sigma^{C}_{0,0|X_0,Y_1}, \nonumber \\
&\sigma^{C}_{0,1|X_1,Y_0}=\sigma^{C}_{0,1|X_0,Y_1}, \nonumber \\
&\sigma^{C}_{1,0|X_1,Y_0}=\sigma^{C}_{1,0|X_0,Y_1}, \nonumber \\
&\sigma^{C}_{1,1|X_1,Y_0}=\sigma^{C}_{1,1|X_0,Y_1}. 
\label{con3}
\end{align}
and
\begin{align}
& \sigma^{C}_{0,0|X_1,Y_1} = \sigma^{C}_{1,0|X_0,Y_0}, \nonumber \\
&\sigma^{C}_{0,1|X_1,Y_1} = \sigma^{C}_{1,1|X_0,Y_0}, \nonumber \\
&\sigma^{C}_{1,0|X_1,Y_1} = \sigma^{C}_{0,0|X_0,Y_0}, \nonumber \\
&\sigma^{C}_{1,1|X_1,Y_1} = \sigma^{C}_{0,1|X_0,Y_0}, 
\label{con4}
\end{align}
Hence, a total of eight different conditional states are produced on Charlie's side, each of which are pure states.

Now, let us assume that the above conditional states (which are pure states) have hybrid LHS description using the ensemble $\{ P(\nu) \, \rho^{C}_{\nu} \}$ and stochastic maps $P(a,b|X_x,Y_y,\nu)$ satisfying Eq.(\ref{assembisep}). It is well-known that a pure state cannot be expressed as a convex sum of other different states, i.e., a density matrix of pure state can only be expanded by itself. Therefore, we can write the following using Eqs.(\ref{assembisep}), (\ref{con1}), (\ref{con3}) and (\ref{con4}),
\begin{align}
&\sigma^{C}_{0,0|X_0,Y_0} =  P(1) \, \rho^{C}_{1}, \nonumber \\
&\sigma^{C}_{0,1|X_0,Y_0} =  P(2) \, \rho^{C}_{2}, \nonumber \\
&\sigma^{C}_{1,0|X_0,Y_0} =  P(3) \, \rho^{C}_{3}, \nonumber \\
&\sigma^{C}_{1,1|X_0,Y_0} =  P(4) \, \rho^{C}_{4}, 
\label{lhs1}
\end{align}
\begin{align}
&\sigma^{C}_{0,0|X_0,Y_1} =  P(5) \, \rho^{C}_{5}, \nonumber \\
&\sigma^{C}_{0,1|X_0,Y_1} =  P(6) \, \rho^{C}_{6}, \nonumber \\
&\sigma^{C}_{1,0|X_0,Y_1} =  P(7) \, \rho^{C}_{7}, \nonumber \\
&\sigma^{C}_{1,1|X_0,Y_1} =  P(8) \, \rho^{C}_{8},
\label{lhs2}
\end{align}
\begin{align}
& \sigma^{C}_{0,0|X_1,Y_0}=P(5) \, \rho^{C}_{5}, \nonumber \\
& \sigma^{C}_{0,1|X_1,Y_0}=P(6) \, \rho^{C}_{6}, \nonumber \\
& \sigma^{C}_{1,0|X_1,Y_0}=P(7) \, \rho^{C}_{7}, \nonumber \\
& \sigma^{C}_{1,1|X_1,Y_0}=P(8) \, \rho^{C}_{8}. 
\label{lhs3}
\end{align}
\begin{align}
&\sigma^{C}_{0,0|X_1,Y_1} = P(3) \, \rho^{C}_{3}, \nonumber \\
&\sigma^{C}_{0,1|X_1,Y_1} = P(4) \, \rho^{C}_{4}, \nonumber \\
& \sigma^{C}_{1,0|X_1,Y_1} = P(1) \, \rho^{C}_{1}, \nonumber \\
&\sigma^{C}_{1,1|X_1,Y_1} = P(2) \, \rho^{C}_{2},
\label{lhs4}
\end{align}
We can therefore, claim that the ensemble $\{ p(\nu) \, \rho^{C}_{\nu} \}$ consists of eight hybrid LHS: $\{P(1) \rho_1^C$, $P(2) \rho_2^C$, $P(3) \rho_3^C$, $P(4) \rho_4^C$, $P(5) \rho_5^C$, $P(6) \rho_6^C$, $P(7) \rho_7^C$, $P(8) \rho_8^C \}$ which reproduces the conditional states $\{\sigma^{C}_{a,b|X_x,Y_y}\}_{a,X_x,b,Y_y}$ at Charlie's end. Now, using Eq.(\ref{marchar}) we can write,
\begin{equation}
\sum_{\nu=1}^{8} P(\nu) \, \rho^{C}_{\nu} = \rho^C_{\text{GGHZ}}.
\label{newww}
\end{equation}

Next, summing Eqs.(\ref{lhs1}), (\ref{lhs2}), (\ref{lhs3}) and (\ref{lhs4}), and then taking trace, the left-hand sides give $4 \tr[\rho^C_{\text{GGHZ}}] = 4$. Here we have used the fact: $\sum_{a=0}^{1} \sum_{b=0}^{1} \sigma^{C}_{a,b|X_x,Y_y} = \rho^C_{\text{GGHZ}}$ $\forall$ $x,y$.
On the other hand, the right-hand sides give $2 \tr[\rho^C_{\text{GGHZ}}] = 2$ following Eq.(\ref{newww}). 
Hence, this leads to a full contradiction of "$2 = 1$".

\begin{enumerate}
	\item \textit{Remark-1} The existence of hybrid LHS models leads to the contradiction $"2 = 1"$ in the 1SDI scenario for any pure GGHZ state when Alice performs dichotomic projective measurements corresponding to the observables: $X_0 = \sigma_x$, $X_1 = \sigma_y$ (see Appendix (\ref{B}) for details).
	\item \textit{Remark-2} The existence of hybrid LHS models leads to the contradiction $"2 = 1"$ in the 1SDI scenario for any pure W-class state when Alice performs the same dichotomic measurements as in case of GGHZ state (see Appendix (\ref{C}) for details).
	\item \textit{Remark-3} The existence of hybrid LHS models leads to the contradiction $"4 = 1"$ in the 2SDI scenario for any pure W-class state when Alice's and Bob's dichotomic projective measurements are same as in case of GGHZ state (see Appendix (\ref{D}) for details).
	\item \textit{Remark-4} The existence of hybrid LHS models leads to no contradiction in both 1SDI and 2SDI scenarios for any pure bi-separable state. However, the existence of  LHS models lead to a contradiction. (see Appendix (\ref{E}) for details).
\end{enumerate}
Using the assemblage decomposition (\ref{lhs}) and similar arguments, it can be shown that the existence of LHS model also lead to contradiction for GGHZ state in our 2SDI scenario. Note that the existence of an LHS model implies the existence of a hybrid LHS model, but not the other way around. This follows from the reasoning that the existence of  LHS models imply the absence of steering which also signify the absence of genuine steering in the multipartite scenario. The above sharp logical contradiction for demonstrating the non-existence of LHS models for the GGHZ states  and W-class states generalizes the EPR paradox to the case of  pure three-qubit entangled states.
Here, Alice's and Bob's two different local measurements prepare different pure conditional states at Charlie's end. 
In Ref. \cite{GBD+18}, it has been demonstrated that  perfect correlations of the EPR paradox can be detected by the 
algebraic maximum of the sum of two conditional probabilities. Similarly, in order to detect the correlation of the GGHZ state demonstrated by the above contradiction, one  may consider the following sum of four conditional probabilities:
\begin{eqnarray}
 CP&:=&P(0_{Z_0}|1_{X_0}1_{Y_0})+P(0_{Z_0}|0_{X_1}1_{Y_1}) \nonumber \\
 &+&P(0_{Z_1}|0_{X_0}1_{{Y_1}})+P(0_{Z_1}|0_{X_1}1_{Y_0}).
 \label{Cp2DI1}
\end{eqnarray}
Here, $P(c_{Z_{z}}|a_{X_x} b_{Y_y})$ denotes the conditional probability of occurrence of the outcome $c$ when Charlie performs measurement $Z_z$, given that Alice and Bob get the outcome $a$ and $b$ by performing measurements $X_x$ and $Y_y$, respectively (with $a$, $b$, $c$, $x$, $y$, $z$ $\in$ $\{0,1\}$).
It can be checked that the GGHZ state gives rise to the algebraic maximum of $4$ for the above quantity for
the following choice of measurements:
\begin{eqnarray}\label{obsGGHZp}
&& X_0=\sigma_x; \quad  Y_0=\sin2\theta \sigma_x +\cos2\theta\sigma_z \quad Z_0=\sigma_x \nonumber \\
&& X_1=\sigma_y;  \quad Y_1=\cos2\theta \sigma_z +\sin2\theta \sigma_y \quad Z_1=\sigma_y
\end{eqnarray}

\section{Fine-grained tripartite steering inequality}\label{section3}

We now  present a fine-grained steering inequality whose violation detects tripartite quantum steering in the 2SDI scenario. The form of the inequality is motivated from the above expression of $CP$ given in Eq.(\ref{Cp2DI1}).

Consider the following two-sided device-independent tripartite scenario:  Alice  performs two arbitrary  black-box dichotomic measurements $X_x$ with $x$ $\in$ $\{0,1\}$ having outcomes $a$ $\in$ $\{0,1\}$. Bob  performs two arbitrary black-box dichotomic measurements $Y_y$ with $y$ $\in$ $\{0,1\}$ having outcomes $b$ $\in$ $\{0,1\}$.  Charlie performs two arbitrary mutually unbiased qubit measurements $Z_z$ with $z$ $\in$ $\{0,1\}$ having outcomes $c$ $\in$ $\{0,1\}$. In the context of this scenario, the tripartite correlation $P(a,b,c|X_x,Y_y, Z_z)$ does not detect tripartite steerability iff it has the following LHV-LHV-LHS decomposition: 
\begin{align}
&P(a,b,c|X_x,Y_y,Z_z) \nonumber \\
&= \sum_{\lambda} P(\lambda) \, P(a|X_x,\lambda) \, P(b|Y_y, \lambda) \, P(c|Z_z, \rho^C_\lambda) 
\label{LHS}
\end{align}

From Eq.(\ref{LHS}), an arbitrary conditional probability distribution can be written as,
\begin{align}
&P(c_{Z_{z}}|a_{X_x} b_{Y_y}) \nonumber \\
&= \frac{\sum_{\lambda} P(\lambda) \, P(a|X_x,\lambda) \, P(b|Y_y, \lambda) \, P(c|Z_z, \rho^C_\lambda)}{P(a,b|X_x,Y_y)} .
\label{Con1}
\end{align}
Now, from the inequality: $\sum_i x_i y_i \leq \big( \max\limits_{i} \lbrace x_i \rbrace \sum_i y_i \big)$ for $x_i \geq 0$ and $y_i \geq 0$, one can write from Eq.(\ref{con1}),
\begin{align} \label{fullyprod}
&P(c_{Z_{z}}|a_{X_x} b_{Y_y}) \nonumber \\
& \leq \max\limits_{\lambda} \big[P(c|Z_z, \rho^C_\lambda) \big]  \Bigg(\frac{\sum_{\lambda} P(\lambda) \, P(a|X_x,\lambda) \, P(b|Y_y, \lambda)}{P(a,b|X_x,Y_y)}\Bigg) \nonumber \\
&= P(c|Z_z, \rho^C_{\lambda_{\text{max}}}),
\end{align}
where we have used $P(a,b|X_x,Y_y) = \sum_{\lambda} P(\lambda) \, P(a|X_x,\lambda) \, P(b|Y_y, \lambda)$ and $P(c|Z_z, \rho^C_{\lambda_{\text{max}}})$ = $\max\limits_{\lambda} \big[P(c|Z_z, \rho^C_\lambda)\big]$. The above inequality is saturated when $\rho^{C}_{\lambda}$ = $\rho^C_{\lambda_{\text{max}}}$ $\forall$ $\lambda$.

Now, let us consider the following sum of conditional probabilities
\begin{align}
\overline{CP} &= P(c_{Z_{0}}|a_{X_0} b_{Y_0}) + P(c_{Z_{1}}|a^{'}_{X_0} b^{'}_{Y_1}) \nonumber \\
&+ P(c_{Z_{0}}|a^{''}_{X_1} b^{''}_{Y_1}) + P(c_{Z_{1}}|a^{'''}_{X_1} b^{'''}_{Y_0}),
\label{cp22}
\end{align}
with $a$, $a^{'}$, $a^{''}$, $a^{'''}$, $b$, $b^{'}$, $b^{''}$, $b^{'''}$, $c$ $\in$ $\{0,1\}$. Note that CP given by Eq.(\ref{Cp2DI1}) is a specific case of $\overline{CP}$.
Since, the trusted party Charlie performs two arbitrary mutually unbiased qubit measurements, following the approach adopted for deriving the fine grained bipartite steering inequality in \cite{PKM14}, one obtains
\begin{align}
\overline{CP} & \leq 2 \max\limits_{\{\tilde{Z}_0, \tilde{Z}_1\}} \big[P(c| \tilde{Z}_0, \rho^C_{\lambda_{\text{max}}}) + P(c| \tilde{Z}_1, \rho^C_{\lambda_{\text{max}}}) \big],
\label{fgi}
\end{align}
where $\{(\tilde{Z}_0, \tilde{Z}_1\}$ ranges over all possible pairs of mutually unbiased qubit measurements. The right hand side of the above inequality measures the uncertainty arising from mutually unbiased qubit measurements $\{\tilde{Z}_0, \tilde{Z}_1 \}$ and is bounded by the fine grained uncertainty relation \cite{OW10}.
 
The task of tripartite quantum steering is demonstrated if  Alice and Bob are able to convince Charlie that their shared state is entangled. Let us discuss Alice-Bob's strategy to cheat Charlie when Charlie's particle is a qubit. Alice and Bob try to maximize the right hand side of inequality (\ref{fgi}) using the LHS model. Here we consider two different scenarios separately \cite{PKM14}. 

In the 1st scenario, Alice and Bob get the information of $\{Z_0, Z_1\}$ before preparing the tripartite state. In this case the following can be shown \cite{PKM14}
\begin{align}
\overline{CP} & \leq 2 (1+\frac{1}{\sqrt{2}}) \nonumber \\
& = 2 + \sqrt{2}.
\label{fgi2}
\end{align}
The above inequality can be derived using the fine-grained uncertainty relation and its violation implies tripartite quantum steering  in the 2SDI scenario.

In the 2nd scenario, Alice and Bob prepare the state without getting the information of $\{Z_0, Z_1\}$. Hence, in this case Alice and Bob prepare all systems without any knowledge of Charlie's set of observables. In this scenario, following the approach adopted in  \cite{PKM14}, it can be shown that 
\begin{align}
\overline{CP} \leq  3.
\label{fgi3}
\end{align}
Quantum violation of the above inequality implies tripartite quantum steering in 2SDI scenario.

Hence, for the expression (\ref{Cp2DI1}), when the shared state does not demonstrate tripartite steering, we can write 
\begin{align} 
 CP&=P(0_{Z_0}|1_{X_0}1_{Y_0})+P(0_{Z_0}|0_{X_1}1_{Y_1}) \nonumber \\
 &\hspace{0.4cm}+ P(0_{Z_1}|0_{X_0}1_{{Y_1}})+P(0_{Z_1}|0_{X_1}1_{Y_0}) \nonumber \\
& \le 2+\sqrt{2} \hspace{1cm} \text{ (1st Scenario)} \nonumber \\
& \le 3 \hspace{1.88cm} \text{ (2nd Scenario)} .
 \label{CP2DI}
\end{align} 
Note that any pure GGHZ state given by Eq.(\ref{GGHZ}) violates the above 2SDI tripartite steering inequality to its algebraic maximum of $4$ for the observables given in Eq. (\ref{obsGGHZp}).


\section{Certification of genuinely entangled three-qubit pure states}\label{section4}

In this section we will derive the main result of the paper regarding certification of genuine tripartite entanglement.
 We will  show that the maximum quantum violation of the above fine-grained inequality (FGI) given by Eq. (\ref{CP2DI}) can be used as a tool for certification of genuinely entangled three-qubit pure states in the 2SDI scenario.
We adopt here a two-step process. At first, we prove that if the shared state is a three-qubit state, the maximum violation of the FGI given by (\ref{CP2DI}) certifies that the state is genuinely entangled pure state. We then show that if the dimension of the shared state is $d_A \times d_B \times 2$,  the maximum violation of the FGI given by (\ref{CP2DI}) certifies that the state is a direct sum of copies of three-qubit genuinely entangled pure states. The analysis presented below is summarized  in the form of the result, stated at the end of this section.

\begin{Lemma}
Suppose that the trusted party, Charlie, performs projective qubit mutually unbiased measurements  corresponding to the operators $Z_{0}=\sigma_x$ and  $Z_{1}=\sigma_y$ and the shared state is a three-qubit state. Then, maximum violation of FGI given by (\ref{CP2DI}) certifies that the three-qubit state is genuinely entangled pure state.
\label{lemma1}
\end{Lemma}

\begin{proof}
Note that the conditional probabilities in this FGI can be written as $P(c_{Z_{z}}|a_{X_x} b_{Y_y})$ = $\tr \left( \Pi_{c|Z_z}	\cdot \varrho^{C}_{a,b|X_x,Y_y} \right)$, where $\Pi_{c|Z_z}$ is the projector associated with the $c$ outcome of $Z_z$ measurement of Charlie. Now, the quantum violation of the FGI becomes $4$, when each of the four conditional probabilities appearing on the left hand side of the FGI given by (\ref{CP2DI}) is $1$. Hence, the following conditions on the normalized conditional states $\{\varrho^{C}_{a,b|X_x,Y_y}\}_{a,X_x,b,Y_y}$ should be satisfied simultaneously when maximum violation ($4$) of FGI is obtained:

$\bullet$ When Alice gets outcome $1$ by measuring $X_0$ and Bob gets outcome $1$ by measuring $Y_0$, then the conditional state prepared at Charlie's end must be eigenstate of the operator $Z_{0}= \sigma_x$ associated with eigenvalue $+1$, i.e.,
\begin{equation}
\varrho^{C}_{1,1|X_0,Y_0} = \frac{\mathbb{I} + \sigma_x}{2}.
\label{connew1}
\end{equation}

$\bullet$ When Alice gets outcome $0$ by measuring $X_1$ and Bob gets outcome $1$ by measuring $Y_1$, then the conditional state prepared at Charlie's end must be the eigenstate of the operator $Z_{0}= \sigma_x$ associated with eigenvalue $+1$, i.e.,
\begin{equation}
\varrho^{C}_{0,1|X_1,Y_1} = \frac{\mathbb{I} + \sigma_x}{2}.
\label{connew2}
\end{equation}

$\bullet$ When Alice gets outcome $0$ by measuring $X_0$ and Bob gets outcome $1$ by measuring $Y_1$, then the conditional state prepared at Charlie's end must be the eigenstate of the operator $Z_{1}=\sigma_y$ associated with eigenvalue $+1$, i.e.,
\begin{equation}
\varrho^{C}_{0,1|X_0,Y_1} = \frac{\mathbb{I} + \sigma_y}{2}.
\label{connew3}
\end{equation}

$\bullet$ When Alice gets outcome $0$ by measuring $X_1$ and Bob gets outcome $1$ by measuring $Y_0$, then the conditional state prepared at Charlie's end must be the eigenstate of the operator $Z_{1}=\sigma_y$ associated with eigenvalue $+1$, i.e.,
\begin{equation}
\varrho^{C}_{0,1|X_1,Y_0} = \frac{\mathbb{I} + \sigma_y}{2}.
\label{connew4}
\end{equation}

Now, it will be shown that no pure three-qubit state without genuine entanglement can provide maximum quantum violation $4$ of FGI (\ref{CP2DI}).

 \textbf{Pure three-qubit states without genuine entanglement}: Any pure three-qubit state without genuine entanglement can be  in one of  the following two categories:

\textit{i) Fully separable states:} Quantum violation of the tripartite steering inequality (\ref{CP2DI}) implies that the shared state is steerable and hence, entangled. Thus, the fully separable states cannot provide maximum quantum violation $4$ of FGI (\ref{CP2DI}).

\textit{ii) Bi-separable states:} Within the bi-separable states, consider a state $|\psi \rangle$ as shown below,
\be \label{BISEP}
\ket{\psi}=\ket{\psi}_A \otimes \ket{\psi}_{BC},
\ee
where $\ket{\psi}_{BC}$ is an arbitrary pure two-qubit entangled state and  $\ket{\psi}_A$ is an arbitrary pure qubit state. Alice  and Bob perform two arbitrary  projective measurements. For bi-separable states of above kind, Alice's particle is not correlated with Bob's and Charlie's particle. Hence, Charlie's measurement outcome cannot depend on Alice's measurement settings and outcomes. Hence, the sum of conditional probabilities in Eq.(\ref{CP2DI}) will take the following form:
\begin{eqnarray}
CP &=& P(0_{Z_0}|1_{Y_0}) + P(0_{Z_0}|1_{Y_1}) + P(0_{Z_1}|1_{Y_1}) \nonumber \\
 &+& P(0_{Z_1}|1_{Y_0}) .
\label{CPBISEP}
\end{eqnarray}
Hence, in order to get each term of Eq.(\ref{CPBISEP}) equal to $1$, Charlie's conditional state should be simultaneous eigenstate of $Z_{0}= \sigma_x$ and  $Z_{1}=\sigma_y$ when Bob obtains the outcome $1$ by performing the measurement $Y_0$. Similarly, Charlie's conditional state should be simultaneous eigenstate of $Z_{0}= \sigma_x$ as  $Z_{1}=\sigma_y$ when Bob obtains the outcome $1$ by performing measurement $Y_1$. But these are not possible. Hence, an arbitrary bi-separable state of the form (\ref{BISEP}) cannot provide maximum violation $4$ of  FGI (\ref{CP2DI}).

Now, consider a bi-separable pure state $|\psi \rangle$ as shown below,
\be \label{BISEP3}
\ket{\psi}=\ket{\psi}_B \otimes \ket{\psi}_{AC},
\ee
where $\ket{\psi}_{AC}$ is an arbitrary pure two-qubit entangled state and  $\ket{\psi}_B$ is an arbitrary pure qubit state. Alice  and Bob perform two arbitrary  projective measurements. For bi-separable states of above kind, Bob's particle is not correlated with Alice's and Charlie's particle. Hence, Charlie's measurement outcome cannot depend on Bob's measurement settings and outcomes. Hence, the sum of conditional probabilities in Eq.(\ref{CP2DI}) will take the following form:
\begin{eqnarray}
CP &=& P(0_{Z_0}|1_{X_0}) + P(0_{Z_0}|0_{X_1}) + P(0_{Z_1}|0_{X_0}) \nonumber \\
 &+& P(0_{Z_1}|0_{X_1}) .
\label{CPBISEP3}
\end{eqnarray}
Hence, in order to get each term of Eq.(\ref{CPBISEP}) equal to $1$, Charlie's conditional state should be simultaneous eigenstate of $Z_{0}= \sigma_x$ as  $Z_{1}=\sigma_y$ when Alice obtains the outcome $0$ by performing measurement $X_1$. But this is not possible. Hence, an arbitrary bi-separable state of the form (\ref{BISEP3}) cannot provide maximum violation $4$ of  FGI (\ref{CP2DI}).

Now, consider the bi-separable states of the following type,
\be \label{BISEP2}
\ket{\psi}=\ket{\psi}_C \otimes \ket{\psi}_{AB},
\ee 
where $\ket{\psi}_{AB}$ is an arbitrary pure two-qubit entangled state and  $\ket{\psi}_C$ is an arbitrary pure qubit state.  In this case Charlie's particle is not correlated with Alice's particle and Bob's particle. Hence, the sum of conditional probabilities in Eq.(\ref{CP2DI}) will take the following form:
\begin{equation}
CP = 2\big[P(0|Z_0, \ket{\psi}_C)+ P(0|Z_1,\ket{\psi}_C) \big] .
\label{CPBISEP2}
\end{equation}
where $P(0|Z_z, \ket{\psi}_C)$ is the probability of occurrence of the outcome $0$ when the measurement of observable $Z_z$ is performed on the state $\ket{\psi}_C$. The above expression of $CP$ can never give quantum violation of  FGI (\ref{CP2DI}) due to the fine-grained uncertainty relation.

 Next, we will show that no mixed three-qubit state (with or without genuine entanglement)  can give maximum quantum violation $4$ of the FGI given by  (\ref{CP2DI}).

\textbf{Mixed three-qubit states}:  Three-qubit mixed states can be classified as follows \citep{ALS1}:

\textit{i) Fully separable states ($\mathbb{S}$):} This class of states includes those states that can be expressed as convex combination of fully separable pure states. The mixed states belonging to this class, being not entangled, never violate the FGI given by  (\ref{CP2DI}).

\textit{ii) Bi-separable states ($\mathbb{B}$):} These are the states that can be expressed as convex combination of fully separable pure states and bi-separable pure states. 

\textit{iii) W-class states ($\mathbb{W}$):} These are the states that can be expressed as convex combination of fully separable pure states, bi-separable pure states and W-class pure states. 

\textit{iv) GHZ class states ($\mathbb{GHZ}$):} These are the states that can be expressed as convex combination of fully separable pure states and bi-separable pure states, W-class pure states and GHZ-class pure states.  

Hence, in general, any three-qubit mixed state can be written as a convex combination of fully separable pure states and bi-separable pure states, W-class pure states and GHZ-class pure states.

Before proceeding, we want to mention that, for any two genuinely entangled three-qubit pure states (W-class pure states or GHZ-class pure states),  FGI (\ref{CP2DI}) does not give maximum quantum violation ($= 4$) for the same set of measurement settings by the two untrusted parties (Detailed numerical proof is given in the Proposition (1) of the appendix \ref{AD}).

Let us consider an arbitrary mixed three-qubit state $\rho_m$. Since any three-qubit mixed state can be expressed as a convex combination of fully separable pure states and bi-separable pure states, W-class pure states and GHZ-class pure states, we can write the following general decomposition of $\rho_m$,
\begin{align}
\rho_m = \sum_{i} p_{i} \rho^i_{fs} + \sum_{j} q_{j} \rho^{j}_{bs}  +  \sum_{k} r_{k} \rho^k_{W} +  \sum_{l} s_{l} \rho^l_{GHZ},
\label{genw}
\end{align}
 where $0 \leq p_i \leq 1$ $\forall$ $i$, $0 \leq q_j \leq 1$ $\forall$ $j$, $0 \leq r_k \leq 1$ $\forall$ $k$, $0 \leq s_l \leq 1$ $\forall$ $l$, $\sum_{i} p_{i} + \sum_{j} q_{j} +  \sum_{k} r_{k}  +  \sum_{l} s_{l}$ = $1$, $\rho^i_{fs}$ is a fully separable three-qubit pure state for all $i$, $\rho^j_{bs}$ is a bi-separable three-qubit  pure state for all $j$, $\rho^k_{W}$ is a three-qubit W-class pure state for all $k$, $\rho^l_{GHZ}$ is a three-qubit GHZ-class pure state for all $l$. 
 
Now, suppose that Alice obtains the outcome $1$ by performing measurement of the observable $X_0$ and Bob obtains the outcome $1$ by performing measurement of the observable $Y_0$ on the above mixed three-qubit state $\rho_m$. Hence, the normalized conditional state prepared at Charlie's side is given by,
\begin{align}
&\varrho^{C}_{1,1|X_0,Y_0} \nonumber \\
&= \frac{\tr_{AB}\Big[\rho_m \big(\Pi_{1|X_0} \otimes \Pi_{1|Y_0} \otimes \mathbb{I} \big)\Big]}{\tr\Big[\rho_m \big(\Pi_{1|X_0} \otimes \Pi_{1|Y_0} \otimes \mathbb{I} \big)\Big]} \nonumber \\
&=\frac{1}{P(1,1|X_0,Y_0)}\tr_{AB}\Big[ \big( \sum_{i} p_{i} \rho^i_{fs} + \sum_{j} q_{j} \rho^{j}_{bs}  +  \sum_{k} r_{k} \rho^k_{W} \nonumber\\
&\, \, \hspace{0.4cm}+  \sum_{l} s_{l} \rho^l_{GHZ} \big) \big(\Pi_{1|X_0} \otimes \Pi_{1|Y_0} \otimes \mathbb{I} \big)\Big],
\label{mixed1}
\end{align}
where $P(1,1|X_0,Y_0)$ = $\tr\Big[\rho_m \big(\Pi_{1|X_0} \otimes \Pi_{1|Y_0} \otimes \mathbb{I} \big)\Big]$ is the probability that Alice gets the outcome $1$ by performing measurement of $X_0$ and Bob gets the outcome $1$ by performing the measurement of $Y_0$ on the state $\rho_m$. Next, we have,
\begin{align}
&\tr_{AB}\Big[ \big( \sum_{i} p_{i} \rho^i_{fs} + \sum_{j} q_{j} \rho^{j}_{bs}  +  \sum_{k} r_{k} \rho^k_{W} +  \sum_{l} s_{l} \rho^l_{GHZ} \big) \nonumber\\
&\, \, \hspace{0.4cm} \big(\Pi_{1|X_0} \otimes \Pi_{1|Y_0} \otimes \mathbb{I} \big)\Big] \nonumber \\
&= \sum_{i} p_{i} \tr_{AB}\Big[  \rho^i_{fs} \big(\Pi_{1|X_0} \otimes \Pi_{1|Y_0} \otimes \mathbb{I} \big) \Big] \nonumber \\
& \hspace{0.4cm} + \sum_{j} q_{j}  \tr_{AB}\Big[  \rho^{j}_{bs} \big(\Pi_{1|X_0} \otimes \Pi_{1|Y_0} \otimes \mathbb{I} \big) \Big] \nonumber \\
& \hspace{0.4cm} + \sum_{k} r_{k}  \tr_{AB}\Big[  \rho^{k}_{W} \big(\Pi_{1|X_0} \otimes \Pi_{1|Y_0} \otimes \mathbb{I} \big) \Big] \nonumber \\
& \hspace{0.4cm} + \sum_{l} s_{l}  \tr_{AB}\Big[  \rho^{l}_{GHZ} \big(\Pi_{1|X_0} \otimes \Pi_{1|Y_0} \otimes \mathbb{I} \big) \Big] \nonumber \\
&=  \sum_{i} p_{i} P^i_{fs}(1,1|X_0,Y_0) \varrho^{i^C}_{{fs}_{1,1|X_0,Y_0}} \nonumber \\
& \hspace{0.4cm}+ \sum_{j} q_{j} P^j_{bs}(1,1|X_0,Y_0) \varrho^{j^C}_{{bs}_{1,1|X_0,Y_0}} \nonumber \\
& \hspace{0.4cm}+ \sum_{k} r_{k} P^k_{W}(1,1|X_0,Y_0) \varrho^{k^C}_{{W}_{1,1|X_0,Y_0}} \nonumber \\
& \hspace{0.4cm}+ \sum_{l} s_{l} P^l_{GHZ}(1,1|X_0,Y_0) \varrho^{l^C}_{{GHZ}_{1,1|X_0,Y_0}} \nonumber \\
\label{mixed2}
\end{align}
where $P^i_{fs}(1,1|X_0,Y_0)$ is the probability that Alice gets the outcome $1$ by performing measurement of $X_0$ and Bob gets the outcome $1$ by performing the measurement of $Y_0$ on the state $\rho^i_{fs}$; $\varrho^{i^C}_{{fs}_{1,1|X_0,Y_0}}$ is the normalized conditional state prepared at Charlie's side in this case; $P^j_{bs}(1,1|X_0,Y_0)$, $P^k_{W}(1,1|X_0,Y_0)$, $P^L_{GHZ}(1,1|X_0,Y_0)$  and $\varrho^{j^C}_{{bs}_{1,1|X_0,Y_0}}$, $\varrho^{k^C}_{{W}_{1,1|X_0,Y_0}}$, $\varrho^{l^C}_{{GHZ}_{1,1|X_0,Y_0}}$ are defined similarly. Hence, from Eqs.(\ref{mixed1}) and (\ref{mixed2}), we get
\begin{align}
&\varrho^{C}_{1,1|X_0,Y_0} \nonumber \\
&= \sum_{i} p_{i} \, \tilde{P}^i_{fs} \, \varrho^{i^C}_{{fs}_{1,1|X_0,Y_0}} + \sum_{j} q_{j} \, \tilde{P}^j_{bs} \, \varrho^{j^C}_{{bs}_{1,1|X_0,Y_0}} \nonumber \\
& \vspace{0.4cm} + \sum_{k} r_{k} \, \tilde{P}^k_{W} \, \varrho^{k^C}_{{W}_{1,1|X_0,Y_0}} + \sum_{l} q_{l} \, \tilde{P}^l_{GHZ} \, \varrho^{l^C}_{{GHZ}_{1,1|X_0,Y_0}},
\label{conv}
\end{align}
where 
\begin{align}
\tilde{P}^i_{fs} &=  \frac{P^i_{fs}(1,1|X_0,Y_0)}{P(1,1|X_0,Y_0)} \nonumber \\
& = \frac{\tr\Big[\rho^i_{fs} \big(\Pi_{1|X_0} \otimes \Pi_{1|Y_0} \otimes \mathbb{I} \big)\Big]}{\tr\Big[\rho_m \big(\Pi_{1|X_0} \otimes \Pi_{1|Y_0} \otimes \mathbb{I} \big)\Big]} \nonumber \\
& \leq 1,
\end{align}
and similarly, 
\begin{align}
&\tilde{P}^j_{bs} =  \frac{P^j_{bs}(1,1|X_0,Y_0)}{P(1,1|X_0,Y_0)} \leq 1, \nonumber \\
&\tilde{P}^k_{W} =  \frac{P^k_{W}(1,1|X_0,Y_0)}{P(1,1|X_0,Y_0)} \leq 1, \nonumber \\
&\tilde{P}^l_{GHZ} =  \frac{P^l_{GHZ}(1,1|X_0,Y_0)}{P(1,1|X_0,Y_0)} \leq 1.
\end{align}
Hence, Eq.(\ref{conv}) represents convex combination of $\varrho^{C}_{1,1|X_0,Y_0}$ in terms of different normalized states. 

When the FGI given by  (\ref{CP2DI}) is maximally violated by the state $\rho_m$ given by Eq.(\ref{genw}), the condition given by Eq.(\ref{connew1}) should be satisfied. In other words, $\varrho^{C}_{1,1|X_0,Y_0}$ should be a pure state and eigenstate of $\sigma_x$. Now, any pure state cannot be written as a convex sum of other different states. Hence, the FGI given by  (\ref{CP2DI}) is maximally violated by the state $\rho_m$  (\ref{genw}) only if each of the states $\varrho^{i^C}_{{fs}_{1,1|X_0,Y_0}}$, $ \varrho^{j^C}_{{bs}_{1,1|X_0,Y_0}}$, $ \varrho^{k^C}_{{W}_{1,1|X_0,Y_0}}$ and $ \varrho^{l^C}_{{GHZ}_{1,1|X_0,Y_0}}$ is the eigenstate of $\sigma_x$, i.e., the normalized conditional state prepared from each of the states $\rho^i_{fs}$, $\rho^{j}_{bs}$, $\rho^{k}_{W}$ and $\rho^{l}_{GHZ}$ satisfies the condition (\ref{connew1}). 

Considering the other three terms appearing on the left hand side of  FGI  (\ref{CP2DI}), it can be shown that the state $\rho_m$  (\ref{genw}) gives maximum quantum violation of FGI (\ref{CP2DI}) only if each of the states  $\rho^i_{fs}$, $\rho^{j}_{bs}$, $\rho^{k}_{W}$ and $\rho^{l}_{GHZ}$ satisfies the conditions (\ref{connew1}), (\ref{connew2}), (\ref{connew3}) and (\ref{connew4}) simultaneously, i.e., each of the states  $\rho^i_{fs}$, $\rho^{j}_{bs}$, $\rho^{k}_{W}$ and $\rho^{l}_{GHZ}$ gives maximum quantum violation ($= 4$) of FGI (\ref{CP2DI}) for the same set of measurement settings performed by the two untrusted parties. However, we have already shown that no pure fully separable and no pure bi-separable three-qubit state can maximally violate FGI  (\ref{CP2DI}). On the other hand, for any two genuinely entangled three-qubit pure states (W-class pure states or GHZ-class pure states),  FGI (\ref{CP2DI}) does not give maximum quantum violation ($= 4$) for the same set
  of measurement settings by the two untrusted parties. 

Hence, no mixed three-qubit state can maximally violate FGI  (\ref{CP2DI}).

Hence, when the shared state is a three-qubit state, then the maximum violation of the FGI (\ref{CP2DI}) certifies that the state is genuinely entangled pure state.
\end{proof}

Next, we will present the following Lemma for the general tripartite state in 2SDI steering scenario having dimension $d_A \otimes d_B \otimes 2$ to complete our proof that the maximum quantum violation of tripartite steering inequality (\ref{CP2DI}) certifies genuine entanglement of three-qubit states in 2SDI scenario.

\begin{Lemma}
 If the maximal violation ($4$) of FGI given by (\ref{CP2DI}) is obtained in our $2$SDI scenario from a tripartite state of dimension $d_A \times d_B \times 2$, then the state of the system can be expressed as a direct sum of copies of three-qubit genuinely entangled pure states.
 \label{lemma2}
\end{Lemma}

\begin{proof}
  Here we use a result  \cite{Mas06,Rabelo2012} which states that given two Hermitian operators $A_0$ and $A_1$ with eigenvalues $\pm 1$
 acting on a Hilbert space $\mathcal{H}$, there is a decomposition of $\mathcal{H}$ as a direct sum of subspaces $\mathcal{H}_i$ of dimension $d\le2$ each, such that both $A_0$ and $A_1$ act within
each $\mathcal{H}_i$, that is, they can be written as $A_0= \oplus_i A_0^i$ and
$A_1= \mathop{\oplus}_{i} A_1^i$, where $A_0^i$ and $A_1^i$ act on  $\mathcal{H}_i$.

 In general, in our steering scenario any shared tripartite state lies in $\mathcal{B}(\mathcal{H}_{A'} \otimes \mathcal{H}_{B'} \otimes  \mathcal{H}_{C})$ where the dimension of $\mathcal{H}_{A'}$ and  $\mathcal{H}_{B'}$ (the untrusted sides) are $d_A$ and $d_B$ (where $d_A$ and $d_B$ are arbitrary) respectively, and the dimension of $\mathcal{H}_{C}$ (the trusted side) is $2$. From the above-mentioned result \cite{Mas06} it follows that  $\mathcal{H}_{A'}$ can be expressed as a direct sum of subspaces $\mathcal{H}^u_A$ of dimension $d\le2$ each. Similarly, $\mathcal{H}_{B'}$ can be expressed as a direct sum of subspaces $\mathcal{H}^v_B$ of dimension $d\le2$ each. Hence,
 \begin{align}
\mathcal{H}_{A'} \otimes \mathcal{H}_{B'} \otimes \mathcal{H}_{C} & = (\oplus_{u,v} \mathcal{H}^u_{A} \otimes \mathcal{H}^v_{B}) \otimes \mathcal{H}_{C} \nonumber \\
&\simeq \oplus_{u,v} (\mathcal{H}^u_{A} \otimes \mathcal{H}^v_{B} \otimes \mathcal{H}_{C}).
\label{HD}
\end{align}

Let us consider that $X_x$ = $\Pi_{0|X_x} - \Pi_{1|X_x}$ with $x$ $\in$ $\{0,1\}$, where $\Pi_{a|X_x}$ ($a$ $\in$ $\{0,1\}$) denotes the projector. Hence, one can write  $\Pi_{a|X_x}$ = $\oplus_{u} \Pi^{u}_{a|X_x}$ where each $\Pi^u_{a|X_x}$ acts on $\mathcal{H}^u_{A}$ for all $a$ and $x$. We also denote $\Pi^{u}$ = $\Pi^{u}_{0|X_x} + \Pi^{u}_{1|X_x}$ the projector on $\mathcal{H}^u_A$. Similarly, consider that $Y_y$ = $\Pi_{0|Y_y} - \Pi_{1|Y_y}$ with $y$ $\in$ $\{0,1\}$. One can write  $\Pi_{b|Y_y}$ = $\oplus_{v} \Pi^{v}_{b|Y_y}$ where each $\Pi^{v}_{b|Y_y}$ acts on $\mathcal{H}^v_{B}$ for all $b$ and $y$. We denote $\Pi^{v}$ = $\Pi^{v}_{0|Y_y} + \Pi^{v}_{1|Y_y}$ the projector on $\mathcal{H}^v_B$. On the other hand, $Z_z$ = $\Pi_{0|Z_z} - \Pi_{1|Z_z}$ with $z$ $\in$ $\{0,1\}$, where $\Pi_{c|Z_z}$ ($c$ $\in$ $\{0,1\}$) denotes the projector acting on $\mathcal{H}_{C}$ of dimension $2$.

 Hence, for any state $\rho$ $\in$ $\mathcal{B}(\mathcal{H}_{A'} \otimes \mathcal{H}_{B'} \otimes  \mathcal{H}_{C})$, we have
\begin{align}
&P(c_{Z_{z}}|a_{X_x} b_{Y_y}) \nonumber \\
 & = \frac{P(a,b,c|X_x,Y_y,Z_z)}{P(a,b|X_x,Y_y)} \nonumber \\
&= \frac{\tr\Big[\rho \big(\Pi_{a|X_x} \otimes \Pi_{b|Y_y} \otimes \Pi_{c|Z_z} \big) \Big]}{\tr\Big[\rho \big(\Pi_{a|X_x} \otimes \Pi_{b|Y_y} \big) \Big]} \nonumber \\
&= \frac{\sum_{u,v} q_{uv} \tr\Big[\rho_{uv} \big(\Pi^u_{a|X_x} \otimes \Pi^v_{b|Y_y} \otimes \Pi_{c|Z_z} \big) \Big]}{\sum_{u,v} q_{uv} \tr\Big[\rho_{uv} \big(\Pi^u_{a|X_x} \otimes \Pi^v_{b|Y_y} \big) \Big]} \nonumber \\
&= \frac{\sum_{u,v} q_{uv} \, P_{uv}(a,b,c|X_x,Y_y,Z_z)}{\sum_{u,v} q_{uv} \, P_{uv}(a,b|X_x,Y_y)},
\end{align}
where $q_{uv}$ = $\tr\Big[\rho \big(\Pi^{u} \otimes \Pi^v \otimes \mathbb{I}\big)\Big]$; $\sum_{u,v} q_{uv} = 1$ and $\rho_{uv}$ = $\dfrac{\big(\Pi^{u} \otimes \Pi^v \otimes \mathbb{I}\big) \rho \big(\Pi^{u} \otimes \Pi^v \otimes \mathbb{I}\big)}{q_{uv}}$ $\in$ $\mathcal{B}(\mathcal{H}^u_{A} \otimes \mathcal{H}^v_{B} \otimes  \mathcal{H}_{C})$ is, at most, a three-qubit state.

Now, $P(c_{Z_{z}}|a_{X_x} b_{Y_y})$ will be equal to $1$ if and only if 
\begin{align}
\sum_{u,v} q_{uv} \big( P_{uv}(a,b|X_x,Y_y) -P_{uv}(a,b,c|X_x,Y_y,Z_z)\big) = 0.
\label{condition1}
\end{align} 
Since for all $\rho_{uv}$, the conditional probability $P_{uv}(c_{Z_{z}}|a_{X_x} b_{Y_y})$ = $\dfrac{P_{uv}(a,b,c|X_x,Y_y,Z_z)}{P_{uv}(a,b|X_x,Y_y)}$ $\leq$ $1$, we have $P_{uv}(a,b|X_x,Y_y) -P_{uv}(a,b,c|X_x,Y_y,Z_z)$ $\geq$ $0$ for all $u$, $v$. On the other hand,  $q_{uv} \geq 0$ for all $u$, $v$.
Hence, the above condition (\ref{condition1}) is satisfied if and only if 
\begin{align}
 q_{uv}  = 0 \hspace{0.5cm} \text{or} \hspace{0.5cm} P_{uv}(c_{Z_{z}}|a_{X_x} b_{Y_y})=1 \hspace{0.5cm} \forall \hspace{0.2cm} u,v.
\label{conditionf}
\end{align}

Next, let us define the following,
\begin{align} 
 CP_{uv}&=P_{uv}(0_{Z_0}|1_{X_0}1_{Y_0})+P_{uv}(0_{Z_0}|0_{X_1}1_{Y_1}) \nonumber \\
 &\hspace{0.4cm}+ P_{uv}(0_{Z_1}|0_{X_0}1_{{Y_1}})+P_{uv}(0_{Z_1}|0_{X_1}1_{Y_0})  .
 \label{CP2DIuv}
 \end{align}
 Now, from the above argument it can be concluded that the maximal violation ($4$) of FGI given by (\ref{CP2DI}) is obtained if and only if $CP_{uv}$ = $4$ for all $u$, $v$ unless $q_{uv}$ = $0$.

Now, from Lemma \ref{lemma1}, it is observed that $CP_{uv}$ = $4$ certifies genuinely entangled pure state when the shared state is a three-qubit state.  Hence, if a state $|\psi \rangle$ of dimension $d_A \times d_B \times 2$ leads to maximum quantum violation ($4$) of FGI given by (\ref{CP2DI}), then it is given by,
\begin{align} \label{pdec}
|\psi \rangle=\oplus_{u,v} \sqrt{q_{uv}} |\phi_{uv}^{GE} \rangle,
\end{align}
where $|\phi_{uv}^{GE} \rangle$ are genuinely entangled three-qubit pure states, $\sum_{u,v} q_{uv}=1$, $q_{uv} \geq 0$ $\forall$ $u,v$. Note that when $q_{uv} = 0$, then according to condition (\ref{conditionf}) the corresponding $CP_{uv}$ may not be equal to $4$. But such $q_{uv}$ does not contribute to $|\psi \rangle$.

Note that the shared tripartite state in more general 2SDI steering scenario can have form $\rho_{ABC} = \rho \oplus \rho'$. Here, $\rho$ is a direct sum of three-qubit genuinely entangled states as shown in Eq. (\ref{pdec}) and $\rho'$ can be a direct sum of three-qubit states such that for each three-qubit constituting state, $P(a, b| X_x, Y_y) = P(a, b, c| X_x, Y_y, Z_z) = 0$. Such combination of states also maximally violates the FGI (\ref{CP2DI}) if the state $\rho$ does. Here, the measurements of the three parties, Alice, Bob and Charlie are such that their action on $\rho'$ does not contribute to the conditional probability terms in the FGI (\ref{CP2DI}). Hence, the certified state is not unique.

\end{proof}

As an example, consider a tripartite GGHZ state in the dimension ($d_A \otimes d_B \otimes 2$) which is a direct sum of three-qubit GGHZ states
\begin{equation} \label{MQV4MPE}
	\ket{\psi_{\text{GGHZ}}} = \oplus_{u,v} \sqrt{q_{uv}} |\phi_{uv}^{GGHZ} \rangle,
\end{equation}
where $$|\phi_{uv}^{GGHZ} \rangle=\cos\theta_{uv} \ket{2u,2v,0} + \sin\theta_{uv} \ket{2u+1,2v+1,1} $$
is a three-qubit GGHZ state acting in a subspace of the  $(d_A \times d_B \times 2)$-dimensional
space where $\ket{\psi_{\text{GGHZ}}}$ has the support. It can be easily checked that the state (\ref{MQV4MPE}) maximally violates the FGI for the measurement settings $X_x=\oplus_{u,v} X^{u,v}_x$, 
$Y_y=\oplus_{u,v} Y^{u,v}_y$, where 
\begin{eqnarray}
&& X^{u,v}_0=\sigma^{u,v}_x; \quad  Y^{u,v}_0=\sin2\theta_{u,v} \sigma^{u,v}_x +\cos2\theta_{u,v}\sigma^{u,v}_z \nonumber \\
&& X^{u,v}_1=\sigma^{u,v}_y;  \quad Y^{u,v}_1=\cos2\theta_{u,v} \sigma^{u,v}_z +\sin2\theta_{u,v} \sigma^{u,v}_y \nonumber
\end{eqnarray}
where $\sigma^{u,v}_x$, $\sigma^{u,v}_y$ and $\sigma^{u,v}_z$ are the Pauli matrices acting in a subspace where 
$|\phi_{uv}^{GGHZ} \rangle$ has the support on the untrusted parties' sides and Charlie, the trusted party as usual performs projective qubit mutually unbiased measurements $\sigma_x$ and $\sigma_y$.\\
Therefore, we can state our result  which follows from the above two Lemmas below:\\

\textbf{Result:} The maximal violation of FGI certifies genuine entanglement in three-qubit pure states in our 2SDI scenario.\\

 The maximum violation of the fine-grained inequality (\ref{CP2DI}) certifies that the state is genuinely entangled. However, not every genuinely entangled state reaches  the quantum bound. In our study, we observe that
\begin{itemize}
	\item Every GGHZ class state (\ref{GGHZ}) which is a subclass of pure GHZ class state (\ref{GHZclass}) maximally violates the FGI (\ref{CP2DI}) for the measurements given by Eq.(\ref{obsGGHZp}).
	\item Out of $10^6$ randomly generated pure W-class states (\ref{wclass}), only 44001 states maximally violate the FGI (\ref{CP2DI}) numerically (See the details in Appendix (\ref{AD})).
	\item Only 6879 states out of $10^6$ randomly generated pure GHZ-class states (\ref{GHZclass}) maximally violate the FGI (\ref{CP2DI}) numerically (See the details in Appendix (\ref{AD})).
\end{itemize}
Hence, not every genuinely entangled state maximally violates the FGI (\ref{CP2DI}). However, there are states in  both classes (pure GHZ class and pure W-class states) that achieve the same quantum bound of the inequality (\ref{CP2DI}). So, maximum violation of FGI is sufficient but not necessary for certifying genuine three-qubit entanglement. Next, we show that even the violation and not the maximum violation of the FGI inequality is sufficient to certify the presence of genuine entanglement in 2SDI steering scenario. This is helpful for the experimental setups where maximal violation may not be achievable due to finite precision of the experimental devices. 

Here, we conjecture that in our steering scenario, the LHS bound ($2+\sqrt{2}$) on the FGI is also the bi-separable bound. This implies that if a state violates the FGI inequality in our steering scenario must have genuine entanglement. Note that in Eq. (\ref{fullyprod}) of our derivation of the LHS bound on the FGI,
for each value of $\lambda$, the joint probability of Alice and Bob need not factorize. 
Thus, the LHS bound also holds for the correlations that can be decomposed  as
\begin{align}
    &P(a,b,c|X_x,Y_y,Z_z) \nonumber \\
    &=\sum_{\lambda} P^{C:AB}(\lambda) \, P(a,b|X_x,Y_y, \lambda) \, P(c|Z_z, \rho^C_\lambda)
\end{align}
We also note that the correlations arising from 
noisy GHZ state $\rho^V_{GHZ}$ given by
\begin{equation}
\rho^V_{GHZ}=V \ket{\psi_{\text{GHZ}}} \bra{\psi_{\text{GHZ}}} + (1-V)
\frac{\mathbb{I}_{8 \times 8}}{8}  
\end{equation}
where $\ket{\psi_{\text{GHZ}}}=\frac{1}{\sqrt{2}}(\ket{000}+\ket{111})$ and $\mathbb{I}_{8 \times 8}$ 
is the identity matrix of dimension $8 \times 8$,
violate our FGI if and only if $V> 1/\sqrt{2}$ for the measurements that give rise to the maximal violation of
our FGI by the GHZ state. In Ref. \cite{Bihalan18}, it has been demonstrated that such correlations for $V \le 1/\sqrt{2}$ can be decomposed 
as 
\begin{align}
    &P(a,b,c|X_x,Y_y,Z_z) \nonumber \\
    &=\sum_{\mu} P^{B:AC}(\mu) \, P_Q(a,c|X_x,Z_z, \mu) \, P(b|Y_y, \mu)
\end{align}
where $P_Q(a,c|X_x,Z_z, \mu)$ is a quantum correlation arising from a bipartite state 
$\rho_{AC}$ of dimension $d_A \times 2$ or 
\begin{align}
    &P(a,b,c|X_x,Y_y,Z_z) \nonumber \\
    &=\sum_{\nu} P^{A:BC}(\nu) \, P(a|X_x, \nu) P_Q(b,c|Y_y,Z_z, \nu) \, 
\end{align}
where $P_Q(b,c|Y_y,Z_z, \nu)$ is a quantum correlation arising from a bipartite state 
$\rho_{BC}$ shared by Bob and Charlie of dimension $d_B \times 2$. Therefore, we conjecture
that the LHS bound of  $2+\sqrt{2}$  on our FGI in our steering scenario also holds for the correlations that can be decomposed 
as 
\begin{align}
    &P(a,b,c|X_x,Y_y,Z_z) \nonumber \\
    &=p_1\sum_{\lambda} P^{C:AB}(\lambda) \, P(a,b|X_x,Y_y, \lambda) \, P(c|Z_z, \rho^C_\lambda) \nonumber \\
    &+p_2\sum_{\mu} P^{B:AC}(\mu) \, P_Q(a,c|X_x,Z_z, \mu) \, P(b|Y_y, \mu) \nonumber \\
    &+p_3\sum_{\nu} P^{A:BC}(\nu) \, P(a|X_x, \nu) P_Q(b,c|Y_y,Z_z, \nu).
    \label{bisep}
\end{align}
where, $P^{C:AB}(\lambda), P^{B:AC}(\mu)$ and $P^{A:BC}(\nu)$ are the probability distributions and $\sum_i p_i = 1$ (i = 1, 2, 3). This indicates that the violation of our FGI implies that the correlations between the three parties cannot be decomposed into the bi-separable form (\ref{bisep}). Hence, it demonstrate genuine tripartite steering
which certifies the genuine tripartite entanglement in a $2$SDI way \cite{CSA+15,Bihalan18}.

\section{Conclusions}\label{section6}

In the present work, we have first demonstrated genuine tripartite EPR steering of arbitrary three-qubit pure GGHZ states and W-class state in partial device independent scenario using a logical argument. In particular, we have shown that the existence of a hybrid LHS model for any three-qubit pure GGHZ state leads to the sharp contradiction: ``$2=1$''. This method rules out the possibility of hybrid LHS models in the tripartite 2SDI scenario more uncompromisingly than the usual steering inequalities. This logical argument has been presented following the well-known GHZ theorem \cite{GHZ,Greenberger1990} which has recently been generalized to the bipartite steering scenario. Our logical contradiction may be regarded as a generalization of the ``steering paradox'' \cite{CSX+16} to more than two parties. 

We have further derived a tripartite 2SDI steering inequality based on the fine-grained uncertainty relation \cite{OW10}. This inequality serves as a generalization of the fine-grained bipartite steering inequality \cite{PKM14}. We have shown that the maximum quantum violation of our  tripartite steering inequality certifies genuine entanglement of pure three-qubit states in the 2SDI scenario. Maximum violation of FGI is associated with genuine tripartite steering since it certifies the presence of genuine entanglement and we conjecture that our LHS bound is also the bi-separable bound in the steering scenario that we have considered. 

Before concluding, it may worth highlighting some possible off-shoots of our present study.
First, practical demonstration of this simple logical contradiction aimed towards showing tripartite steering by photon entanglement based experiments should not be difficult to implement. Note that the quantum violation of the bipartite FGI  has been demonstrated experimentally using two-photon polarization-entangled states \cite{Adeline18,BMJ+20}. This opens up the possibility of experimental demonstration of quantum violation of our proposed steering inequality and certifying genuine entanglement in semi-DI scenario based on the FGI in the near future. Finally,  our analysis brings into focus the question as to whether multipartite quantum steering for more than three parties having arbitrary local dimensions and genuine quantum steering can be demonstrated using sharp logical contradiction.

\section*{Acknowledgement} 
C.J. acknowledges S. N. Bose Centre, Kolkata, for the postdoctoral fellowship and the Polish Academy of Sciences, Poland for partial financial support and thanks J. Kaniewski for useful discussions. D.D. acknowledges the Science and Engineering Research Board (SERB), Government of India for financial support through the National Post Doctoral Fellowship (NPDF) (File No. PDF/2020/001358). S.D. acknowledge financial support from the INSPIRE program, Department of Science and Technology, Government of India (Grant No. C/5576/IFD/2015-16). S.G. acknowledges the S. N. Bose National Centre for Basic Sciences, Kolkata for financial support. A.S.M. acknowledges Project No. DST/ICPS/QuEST/2018/98 from the Department of Science and Technology, Government of India. \textbf{On behalf of all authors, the corresponding author states that there is no conflict of interest.}

\begin{widetext}
\appendix

\section{All-versus-nothing proof of genuine steering of GGHZ states in 1SDI scenario}{\label{B}}
Here, we demonstrate that the existence of hybrid LHS models leads to the contradiction "$2 = 1$"  in the 1SDI scenario for  any pure state that belongs the generalized GHZ (Greenberger-Horne-Zeilinger) class (\ref{GGHZ}). Let us recapitulate the form of the assemblage as described by  hybrid LHS models in the 1SDI scenario. In this case, Alice being the untrusted party performs the measurement and the tripartite state imposes the constraints on the observed assemblage at the Bob-Charlie end. If the shared state is of the form (\ref{bisepstate}), the assemblage has the following form \cite{CSA+15}
\begin{align}
	\sigma_{a|X_x}^{BC} = \sum_\lambda p_\lambda^{A':BC} p(a|X_x,\lambda)\rho_\lambda^{BC} + \sum_\mu p_\mu^{B:A'C} \rho_\mu^B \otimes \sigma_{a|X_x,\mu}^C + \sum_\nu p_\nu^{A'B:C} \sigma_{a|X_x,\nu}^B \otimes \rho_\nu^C
\label{assembisep1S}	
\end{align}
The assemblage (\ref{assembisep1S}) contains three terms. The first term is an unsteerable assemblage from Alice to Bob-Charlie.  Bob-Charlie's assemblage is dependent on Alice's input, and output through the common variable $\lambda$. The second term is unsteerable from Alice to Bob but not necessarily from Alice to Charlie. In this case,  Bob-Charlie's assemblage is dependent on Alice's input, output at Charlie's end only, and the common hidden variable $\mu$. The third term is unsteerable from Alice to Charlie but not necessarily from Alice to Bob. In this case,  Bob-Charlie's assemblage is dependent on Alice's input, output at Bob's end only, and the common hidden variable $\nu$. When each element of Bob-Charlie's assemblage  can be written in the form (\ref{assembisep1S}), then the assemblage does not demonstrate genuine EPR steering in 1SDI scenario but it may demonstrate steering. On the other hand, if the assemblage (\ref{assem2SDI}) can be written in the form ($\sum_\lambda P(\lambda)P(a|X_x,\lambda)\rho_\lambda^B \otimes \rho_\lambda^C$), then the assemblage is unsteerable and no signature of steering is present in the 1SDI scenario.

 Now, Alice performs dichotomic projective measurements corresponding to the observables: $X_0 = \sigma_x$ and $X_1 = \sigma_y$ on her part of the shared GGHZ state (\ref{GGHZ}). After Alice's measurements, a total of four unnormalized conditional states $\sigma_{a|X_x}^{BC^{\text{GGHZ}}}$ (with a, x $\in \{0, 1\}$) are prepared at Bob-Charlie's end as mentioned below.
\begin{eqnarray}
	\sigma_{0|X_0}^{BC^{\text{GGHZ}}} &=& \frac{1}{2} \Big| \theta_+^0 \Big\rangle \Big\langle \theta_+^0 \Big| = p(1) \rho_1^{BC},\quad \hspace{0.45cm} \sigma_{1|X_0}^{BC^{\text{GGHZ}}} = \frac{1}{2} \Big| \theta_-^0 \Big\rangle \Big\langle \theta_-^0 \Big| = p(2) \rho_2^{BC}, \nonumber \\ 
    \sigma_{0|X_1}^{BC^{\text{GGHZ}}} &=& \frac{1}{2} \Big| \theta_-^1 \Big\rangle \Big\langle \theta_-^1 \Big| = p(3) \rho_3^{BC},\quad \hspace{0.45cm} \sigma_{1|X_1}^{BC^{\text{GGHZ}}} = \frac{1}{2} \Big| \theta_+^1 \Big\rangle \Big\langle \theta_+^1 \Big| = p(4) \rho_4^{BC},
\label{GGHZassemblage}
\end{eqnarray}
where $\Big| \theta_\pm^0 \Big\rangle = \cos\theta \ket{00} \pm \sin\theta \ket{11}$ and $\Big| \theta_\pm^1 \Big\rangle = \cos\theta \ket{00} \pm i \, \sin\theta \ket{11}$. Hence, a total of four different conditional states are produced on Bob-Charlie's end, each of which are pure states. Since the conditional states are pure, the assemblage is not the convex combination of the three terms of Eq. (\ref{assembisep1S}) but any one of the terms of Eq. (\ref{assembisep1S}). We find that the dependence of  Bob and Charlie's assemblage on Alice's input and output may come from the common hidden variable and it is not the case that only Bob's or Charlie's state changes by the Alice's input and output. This is the feature of the first term of  Eq, (\ref{assembisep1S}). Hence, if the conditional states have hybrid LHS description, then there exists an assemblage $\{p(\lambda) \rho_{\lambda}^{BC}\}$ such that
\begin{equation}
	\sum_{\lambda} p(\lambda) \rho_{\lambda}^{BC} = \rho_{\text{GGHZ}}^{BC} = \text{Tr}_{A} \Big[\big|\psi(\theta)_{\text{GGHZ}} \big\rangle \big\langle \psi(\theta)_{\text{GGHZ}} \big| \Big]
	\label{LHS_1S_GGHZ}
\end{equation}
It is well-known that a pure state cannot be expressed as a convex sum of other different states, i.e., a density matrix of pure state can only be expanded by itself. We can therefore, claim that the ensemble $\{p(\lambda) \rho_{\lambda}^{BC}\}$ consists of four hybrid LHS:\\
 $\{p(1)\rho_1^{BC}, p(2)\rho_2^{BC}, p(3)\rho_3^{BC}, p(4)\rho_4^{BC}\}$ which reproduces the conditional states $\{\sigma_{a|X_x}^{BC^{\text{GGHZ}}}\}_{a, X_x}$ at Bob-Charlie's end. Now, using Eq.(\ref{LHS_1S_GGHZ}) we can write,
\begin{equation}
	\sum_{\lambda = 1}^{4} p(\lambda) \rho_{\lambda}^{BC} = \rho_{\text{GGHZ}}^{BC}
	\label{lhs_1S_GGHZ}
\end{equation}
Next, summing Eq.(\ref{GGHZassemblage}) and then taking trace, the left-hand sides give $2 \tr[\rho^{BC}_{\text{GGHZ}}] = 2$. Here we have used the fact: $\sum_{a=0}^{1} \sigma^{BC^{\text{GGHZ}}}_{a|X_x} = \rho^{BC}_{\text{GGHZ}}$ $\forall$ $x$.
On the other hand, the right-hand sides give $ \tr[\rho^{BC}_{\text{GGHZ}}] = 1$ following Eq.(\ref{lhs_1S_GGHZ}). 
Hence, this leads to a full contradiction of "$2 = 1$". Similarly in 2SDI scenario, the existence of hybrid LHS model leads to a contradiction of "$2 = 1$" when the shared state is a three-qubit GGHZ state (\ref{GGHZ}) as shown in Section (\ref{section2}). 

\section{All-versus-nothing proof of genuine steering of W-class states in 1SDI scenario}{\label{C}}
Here, we demonstrate that the existence of hybrid LHS models leads to the contradiction "$2 = 1$"  in the 1SDI scenario for any pure state that belongs the W-class ($|\psi_w \rangle = c_0 \ket{001} + c_1 \ket{010} + \sqrt{1-c_0^2 - c_1^2}\ket{100}$). Alice performs dichotomic projective measurements corresponding to the observables: $X_0 = \sigma_x$ and $X_1 = \sigma_y$. After Alice's measurements, a total of four unnormalized conditional states $\sigma_{a|X_x}^{BC^{\text{W}}}$ (with a, x $\in \{0, 1\}$) are prepared at Bob-Charlie's end as mentioned below.
\begin{eqnarray}
	\sigma_{0|X_0}^{BC^{\text{W}}} &=& \frac{1}{2} \Big|w_+^0 \Big\rangle \Big\langle w_+^0 \Big| = p(1)\rho_{1}^{BC}, \quad  \hspace{0.87cm}		
	\sigma_{1|X_0}^{BC^{\text{W}}} = \frac{1}{2} \Big|w_-^0 \Big\rangle \Big\langle w_-^0 \Big| = p(2)\rho_{2}^{BC},  \nonumber \\
	\sigma_{0|X_1}^{BC^{\text{W}}} &=& \frac{1}{2} \Big|w_+^1 \Big\rangle \Big\langle w_+^1 \Big| = p(3)\rho_{3}^{BC}, \quad 	\hspace{0.87cm}	
	\sigma_{1|X_1}^{BC^{\text{W}}} = \frac{1}{2} \Big|w_-^1 \Big\rangle \Big\langle w_-^1 \Big| = p(4)\rho_{4}^{BC}, 
	\label{wassemblage}
\end{eqnarray}
where, $\Big|w_{\pm}^0\Big\rangle := \sqrt{1-c_0^2-c_1^2} \ket{00} \pm c_0 \ket{01} \pm c_1 \ket{10}$ and  $\Big|w_{\pm}^1\Big\rangle := \sqrt{1-c_0^2-c_1^2} \ket{00} \pm i \, c_0 \ket{01} \pm i \, c_1 \ket{10}$. Hence, a total of four different conditional states are produced on Bob-Charlie's end, each of which are pure states. Since the conditional states are pure, the assemblage is not the convex combination of the three terms of Eq. (\ref{assembisep1S}) but any one of the terms of Eq. (\ref{assembisep1S}). We find that the dependence of Bob and Charlie's assemblage on Alice's input and output may come from the common hidden variable and it is not the case that only Bob's or Charlie's state changes by the Alice's input and output. This is the feature of the first term of  Eq. (\ref{assembisep1S}). Hence,
if the conditional states have hybrid LHS description, then there exists an assemblage $\{p(\lambda) \rho_{\lambda}^{BC}\}$ such that
\begin{equation}
	\sum_{\lambda} p(\lambda) \rho_{\lambda}^{BC} = \rho_{w}^{BC} = \text{Tr}_{A} \Big[\big|\psi_w \big\rangle \big\langle \psi_w \big| \Big]
	\label{LHS_1S_W}
\end{equation}
It is well-known that a pure state cannot be expressed as a convex sum of other different states, i.e., a density matrix of pure state can only be expanded by itself. We can therefore, claim that the ensemble $\{p(\lambda) \rho_{\lambda}^{BC}\}$ consists of four hybrid LHS:\\
 $\{p(1)\rho_1^{BC}, p(2)\rho_2^{BC}, p(3)\rho_3^{BC}, p(4)\rho_4^{BC}\}$ which reproduces the conditional states $\{\sigma_{a|X_x}^{BC^{\text{W}}}\}_{a, X_x}$ at Bob-Charlie's end. Now, using Eq.(\ref{LHS_1S_W}) we can write,
\begin{equation}
	\sum_{\lambda = 1}^{4} p(\lambda) \rho_{\lambda}^{BC} = \rho_{w}^{BC}
	\label{lhs_1S_W}
\end{equation}
Next, summing Eq.(\ref{wassemblage}) and then taking trace, the left-hand sides give $2 \tr[\rho^{BC}_{w}] = 2$. Here we have used the fact: $\sum_{a=0}^{1} \sigma^{BC^{w}}_{a|X_x} = \rho^{BC}_{w}$ $\forall$ $x$.
On the other hand, the right-hand sides give $ \tr[\rho^{BC}_{w}] = 1$ following Eq.(\ref{lhs_1S_W}). 
Hence, this leads to a full contradiction of "$2 = 1$".

\section{All-versus-nothing proof of genuine steering of W-class states in 2SDI scenario}{\label{D}}
Here, we demonstrate that the existence of hybrid LHS models leads to the contradiction "$4 = 1$"  in the 2SDI scenario for any pure state that belongs the W-class ($|\psi_w \rangle = c_0 \ket{001} + c_1 \ket{010} + \sqrt{1-c_0^2 - c_1^2}\ket{100}$). Alice performs dichotomic projective measurements corresponding to the observables: $X_0 = \sigma_x$ and $X_1 = \sigma_y$. On the other hand, Bob performs dichotomic projective measurements corresponding to the observables: $Y_0 = \frac{\sigma_x + \sigma_z}{\sqrt{2}}$ and $Y_1 = \frac{\sigma_y + \sigma_z}{\sqrt{2}}$. After Alice's and Bob's measurements, a total of sixteen unnormalized conditional states $\sigma_{a,b|X_x,Y_y}^{C^{\text{W}}}$ (with a, b, x, y $\in \{0, 1\}$) are prepared at Charlie's end as mentioned below.
{\tiny
\begin{eqnarray}
	\sigma_{0, 0|X_0, Y_0}^{C^{\text{W}}} &=& \frac{N_{w_1}^{00100}}{8} \Big| \psi_{w_1} \Big\rangle^{00100} \Big\langle \psi_{w_1} \Big| = p(1) \rho_1^{C},\quad \hspace{0.45cm} \sigma_{0,1|X_0, Y_0}^{C^{\text{W}}} = \frac{N_{w_1}^{11001}}{8} \Big| \psi_{w_1} \Big\rangle^{11001} \Big\langle \psi_{w_1} \Big| = p(2) \rho_2^{C}, \nonumber \\ 
    \sigma_{1, 0|X_0, Y_0}^{C^{\text{W}}} &=& \frac{N_{w_1}^{00110}}{8} \Big| \psi_{w_1} \Big\rangle^{00110} \Big\langle \psi_{w_1} \Big| = p(3) \rho_3^{C},\quad \hspace{0.45cm} \sigma_{11|X_0,Y_0}^{C^{\text{W}}} = \frac{N_{w_1}^{11011}}{8} \Big| \psi_{w_1} \Big\rangle^{11011} \Big\langle \psi_{w_1} \Big| = p(4) \rho_4^{C}, \nonumber \\ 
     \sigma_{0, 0|X_0, Y_1}^{C^{\text{W}}} &=& \frac{N_{w_2}^{01100}}{8} \Big| \psi_{w_2} \Big\rangle^{01100} \Big\langle \psi_{w_2} \Big| = p(5) \rho_5^{C},\quad \hspace{0.45cm} \sigma_{0,1|X_0, Y_1}^{C^{\text{W}}} = \frac{N_{w_2}^{10001}}{8} \Big| \psi_{w_2} \Big\rangle^{10001} \Big\langle \psi_{w_2} \Big| = p(6) \rho_6^{C}, \nonumber \\ 
    \sigma_{1, 0|X_0, Y_1}^{C^{\text{W}}} &=& \frac{N_{w_2}^{01110}}{8} \Big| \psi_{w_2} \Big\rangle^{01110} \Big\langle \psi_{w_2} \Big| = p(7) \rho_7^{C},\quad \hspace{0.45cm} \sigma_{11|X_0,Y_1}^{C^{\text{W}}} = \frac{N_{w_2}^{10011}}{8} \Big| \psi_{w_2} \Big\rangle^{10011} \Big\langle \psi_{w_2} \Big| = p(8) \rho_8^{C}, \nonumber \\ 
    \sigma_{0, 0|X_1, Y_0}^{C^{\text{W}}} &=& \frac{N_{w_3}^{00110}}{8} \Big| \psi_{w_3} \Big\rangle^{00110} \Big\langle \psi_{w_3} \Big| = p(9) \rho_9^{C},\quad \hspace{0.45cm} \sigma_{0,1|X_1, Y_0}^{C^{\text{W}}} = \frac{N_{w_3}^{11011}}{8} \Big| \psi_{w_3} \Big\rangle^{11011} \Big\langle \psi_{w_3} \Big| = p(10) \rho_{10}^{C}, \nonumber \\ 
    \sigma_{1, 0|X_1, Y_0}^{C^{\text{W}}} &=& \frac{N_{w_3}^{00100}}{8} \Big| \psi_{w_3} \Big\rangle^{00100} \Big\langle \psi_{w_3} \Big| = p(11) \rho_{11}^{C},\quad \hspace{0.45cm} \sigma_{11|X_0,Y_0}^{C^{\text{W}}} = \frac{N_{w_3}^{11001}}{8} \Big| \psi_{w_3} \Big\rangle^{11001} \Big\langle \psi_{w_3} \Big| = p(12) \rho_{12}^{C}, \nonumber \\ 
    \sigma_{0, 0|X_1, Y_1}^{C^{\text{W}}} &=& \frac{N_{w_4}^{01100}}{8} \Big| \psi_{w_4} \Big\rangle^{01100} \Big\langle \psi_{w_4} \Big| = p(13) \rho_{13}^{C},\quad \hspace{0.45cm} \sigma_{0,1|X_1, Y_1}^{C^{\text{W}}} = \frac{N_{w_4}^{10011}}{8} \Big| \psi_{w_4} \Big\rangle^{10011} \Big\langle \psi_{w_4} \Big| = p(14) \rho_{14}^{C}, \nonumber \\ 
    \sigma_{1, 0|X_1, Y_1}^{C^{\text{W}}} &=& \frac{N_{w_4}^{01110}}{8} \Big| \psi_{w_4} \Big\rangle^{01110} \Big\langle \psi_{w_4} \Big| = p(15) \rho_{15}^{C},\quad \hspace{0.45cm} \sigma_{11|X_1,Y_1}^{C^{\text{W}}} = \frac{N_{w_4}^{10001}}{8} \Big| \psi_{w_4} \Big\rangle^{10001} \Big\langle \psi_{w_4} \Big| = p(16) \rho_{16}^{C},
\label{W2assemblage}
\end{eqnarray}}
{\tiny
where, $\Big| \psi_{w_1}^{abcde} \Big\rangle = \dfrac{\sqrt{2+(-1)^a\sqrt{2}}c_0\ket{1}+(-1)^b\sqrt{2+(-1)^c\sqrt{2}}c_1\ket{0}+(-1)^d\sqrt{2+(-1)^e\sqrt{2}}\sqrt{1-c_0^2-c_1^2}\ket{0}}{\sqrt{N_{w_1}^{abcde}}}$,\\
$\Big| \psi_{w_2}^{abcde} \Big\rangle = \dfrac{\sqrt{2+(-1)^a\sqrt{2}}c_0\ket{1}+(-1)^b \iota \sqrt{2+(-1)^c\sqrt{2}}c_1\ket{0}+(-1)^d\sqrt{2+(-1)^e\sqrt{2}}\sqrt{1-c_0^2-c_1^2}\ket{0}}{\sqrt{N_{w_2}^{abcde}}}$,\\
$\Big| \psi_{w_3}^{abcde} \Big\rangle = \dfrac{\sqrt{2+(-1)^a\sqrt{2}}c_0\ket{1}+(-1)^b\sqrt{2+(-1)^c\sqrt{2}}c_1\ket{0}+(-1)^d \iota \sqrt{2+(-1)^e\sqrt{2}}\sqrt{1-c_0^2-c_1^2}\ket{0}}{\sqrt{N_{w_3}^{abcde}}}$,\\
$\Big| \psi_{w_4}^{abcde} \Big\rangle = \dfrac{\sqrt{2+(-1)^a\sqrt{2}}c_0\ket{1}+(-1)^b \iota (\sqrt{2+(-1)^c\sqrt{2}}c_1\ket{0}+(-1)^d\sqrt{2+(-1)^e\sqrt{2}}\sqrt{1-c_0^2-c_1^2}\ket{0})}{\sqrt{N_{w_4}^{abcde}}}$.\\ }
Hence, a total of sixteen different conditional states are produced on Charlie's end, each of which are pure states. Since the conditional states are pure, the assemblage is not the convex combination of the three terms of Eq. (\ref{assembisep}) but any one of the terms of Eq. (\ref{assembisep}). We find that the dependence of  Charlie's assemblage on Alice and Bob's input and output may come from the common hidden variable and it is not the case that either Alice or Bob can change Charlie's state through their input and output choices. This is the feature of the third term of the Eq. (\ref{assembisep}). Hence, if the conditional states have hybrid LHS description, then there exists an assemblage $\{p(\nu) \rho_{\nu}^{C}\}$ such that
\begin{equation}
	\sum_{\nu} p(\nu) \rho_{\nu}^{C} = \rho_{w}^{C} = \text{Tr}_{AB} \Big[\big|\psi_w \big\rangle \big\langle \psi_w \big| \Big]
	\label{LHS_2S_W}
\end{equation}
It is well-known that a pure state cannot be expressed as a convex sum of other different states, i.e., a density matrix of pure state can only be expanded by itself. We can therefore, claim that the ensemble $\{p(\nu) \rho_{\nu}^{C}\}$ consists of sixteen hybrid LHS: $\{p(1)\rho_1^{C}, p(2)\rho_2^{C}, p(3)\rho_3^{C}, p(4)\rho_4^{C}$, $p(5)\rho_5^{C}, 
p(6)\rho_6^{C}, p(7)\rho_7^{C}, p(8)\rho_8^{C}, p(9)\rho_9^{C}, p(10)\rho_{10}^{C}, p(11)\rho_{11}^{C}, p(12)\rho_{12}^{C}, p(13)\rho_{13}^{C}, p(14)\rho_{14}^{C},$
$ p(15)\rho_{15}^{C}$, $p(16)\rho_{16}^{C}\}$ which reproduces the conditional states $\{\sigma_{a, b|X_x, Y_y}^{C}\}_{a, b, X_x, Y_y}$ at Charlie's end. Now, using Eq.(\ref{LHS_2S_W}) we can write,
\begin{equation}
	\sum_{\nu = 1}^{16} p(\nu) \rho_{\nu}^{C} = \rho_{w}^{C}
	\label{lhs_2S_w}
\end{equation}
Next, summing Eq.(\ref{W2assemblage}) and then taking trace, the left-hand sides give $4 \tr[\rho^{C}_w] = 4$. Here we have used the fact: $\sum_{a ,b = 0}^{1} \sigma^{C}_{a, b|X_x, Y_y} = \rho^{C}_{w}$ $\forall$ $x, y$.
On the other hand, the right-hand sides give $ \tr[\rho^{C}_{w}] = 1$ following Eq.(\ref{lhs_2S_w}). 
Hence, this leads to a full contradiction of "$4 = 1$".

\section{All-versus-nothing proof of steering of pure bi-separable states in 1SDI and 2SDI scenario}{\label{E}}
We first demonstrate that the existence of a hybrid LHS model leads to no contradiction in the 1SDI scenario for bi-separable state, but the existence of LHS models may lead to a contradiction. Consider a state $\psi$ of the form (\ref{BISEP2}). In particular, consider the following bi-separable state,
\begin{equation}
	\psi_{bs} = (\cos \theta_1 \ket{00}_{AB}+\sin \theta_1 \ket{11}_{AB}) \otimes (\cos \theta_2 \ket{0}_C+\sin \theta_2 \ket{1}_C)
\label{bs}
\end{equation} 
Alice performs dichotomic projective measurements corresponding to the observables: $X_0 = \sigma_x$ and $X_1 = \sigma_y$ on her part of the shared bi-separable state (\ref{bs}). After Alice's measurements, a total of four unnormalized conditional states $\sigma_{a|X_x}^{BC^{\text{bs}}}$ (with a, x $\in \{0, 1\}$) are prepared at Bob-Charlie's end as mentioned below.
{\small
\begin{eqnarray}
	\sigma_{0|X_0}^{BC^{\text{bs}}} &=& \frac{1}{2} \Big| \theta_{++}^0 \Big\rangle \Big\langle \theta_{++}^0 \Big| = p(1) \sigma_{0|X_0,1}\otimes \rho_1^{C},\quad \hspace{0.45cm} \sigma_{1|X_0}^{BC^{\text{bs}}} = \frac{1}{2} \Big| \theta_{-+}^0 \Big\rangle \Big\langle \theta_{-+}^0 \Big| = p(1)\sigma_{1|X_0,1}\otimes \rho_1^{C}, \nonumber \\ 
    \sigma_{0|X_1}^{BC^{\text{bs}}} &=& \frac{1}{2} \Big| \theta_{-+}^1 \Big\rangle \Big\langle \theta_{-+}^1 \Big| = p(2) \sigma_{0|X_1,2}\otimes \rho_2^{C},\quad \hspace{0.45cm} \sigma_{1|X_1}^{BC^{\text{bs}}} = \frac{1}{2} \Big| \theta_{++}^1 \Big\rangle \Big\langle \theta_{++}^1 \Big| = p(2) \sigma_{1|X_1,2} \otimes \rho_2^{C},
\label{bsassemblage1S}
\end{eqnarray}}
where $\Big| \theta_{\pm,\pm}^0 \Big\rangle = (\cos\theta_1 \ket{0} \pm \sin\theta_1 \ket{1}) \otimes (\cos\theta_2 \ket{0} \pm \sin\theta_2 \ket{1})$ and $\Big| \theta_{\pm,\pm}^1 \Big\rangle = (\cos\theta_1 \ket{0} \pm \iota \sin\theta_1 \ket{1}) \otimes (\cos\theta_2 \ket{0} \pm \sin\theta_2 \ket{1})$. Hence, a total of two distinguishable conditional states are produced at Bob-Charlie's end, all of which are pure states. Since the conditional states are pure, the assemblage is not a convex combination of the three terms of Eq. (\ref{assembisep1S}) but any one of the terms of Eq. (\ref{assembisep1S}). We find that the dependence of Bob and Charlie's assemblage on Alice's input and output may come from the common hidden variable and it is the case that only Bob's state changes by the Alice's input and output. This is the feature of third term of  Eq. (\ref{assembisep1S}). The common variable for the conditional state $\sigma_{1|X_0}^{BC^{bs}}$ is same as that of $\sigma_{0|X_0}^{BC^{bs}}$. So, using a hybrid LHS model we cannot distinguish between them. Similarly, the states $\sigma_{0|X_1}^{BC^{bs}}$ and $\sigma_{1|X_1}^{BC^{bs}}$ are same according to the hybrid LHS model. Now,
if the conditional states have hybrid LHS description, then there exists an assemblage $\{p(\nu) \rho_\nu^B \otimes \rho_{\nu}^{C}\}$ such that
\begin{equation}
	\sum_{\nu} p(\nu) \rho_{\nu}^{B} \otimes \rho_{\nu}^{C} = \rho_{\text{bs}}^{BC} = \text{Tr}_{A} \Big[\big|\psi_{\text{bs}} \big\rangle \big\langle \psi_{\text{bs}} \big| \Big]
	\label{LHS_1S_BS}
\end{equation}
So, according to the hybrid LHS model, the ensemble $\{p(\nu) \rho_{\nu}^{B} \otimes \rho_{\nu}^{C}\}$ consists of four conditional states out of which two are distinguishable. These two reproduce the conditional states $\{\sigma_{a|X_x}^{BC^{\text{bs}}}\}_{a, X_x}$ at Bob-Charlie's end. Now, using Eq.(\ref{LHS_1S_BS}) we can write,
\begin{equation}
	\sum_{\nu = 1}^{2} p(\nu) \rho_{\nu}^{B} \otimes \rho_{\nu}^{C} = \rho_{\text{bs}}^{BC}
	\label{lhs_1S_BS}
\end{equation}
Next, summing Eq.(\ref{bsassemblage1S}) and then taking trace, both the left-hand  and right-hand sides give $2 \tr[\rho^{BC}_{\text{GGHZ}}] = 2$. Here, we have used the fact: $\sum_{a=0}^{1} \sigma^{BC^{\text{bs}}}_{a|X_x} = \rho^{BC}_{\text{bs}}$ $\forall$ $x$ and $\tr[\rho^{BC}_{\text{bs}}] = 1$ following Eq.(\ref{lhs_1S_BS}). 
Hence, in this case there is no contradiction.
\begin{enumerate}
	\item \textit{Remark-1} Note that in case of an  LHS model, there will be four distinct conditional states and it leads to the contradiction $"2 = 1"$ in the 1SDI scenario for the bi-separable states of the form (\ref{bs}) when Alice performs dichotomic projective measurements corresponding to the observables: $X_0 = \sigma_x$, $X_1 = \sigma_y$. This demonstrates steering in such states.
	\item \textit{Remark-2} The existence of LHS or hybrid LHS  models lead to NO contradiction in the 1SDI scenario for bi-separable states of the form (\ref{BISEP}).
	\item \textit{Remark-3} The characteristics of the pure bi-separable states in which Alice and Charlie are entangled are same as the state of the form (\ref{bs}). This means no contradiction occurs for hybrid LHS models but the existence of LHS models lead to a contradiction.
	\item \textit{Remark-4} Following similar reasoning and using a hybrid LHS model of the form (\ref{assembisep}), it can be shown that the existence of  LHS models lead to no contradiction in the 2SDI scenario for bi-separable states, but the existence of LHS models may lead to a contradiction.
\end{enumerate}

\section{}{\label{AD}}
\begin{proposition}
For any two genuinely entangled three-qubit pure states (W-class pure states or GHZ-class pure states),  FGI (\ref{CP2DI}) does not give maximum quantum violation ($= 4$) for the same set of measurement settings by the two untrusted parties.
\end{proposition}
\begin{proof} 
We use the following numerical strategy to show that no two pure genuinely entangled states can give rise to the maximum violation of FGI (\ref{CP2DI}) for the same set of measurement settings by the untrusted parties:\\\\
\textbf{Numerical strategy:} The precision is set at the 6th decimal place for the numerical calculations. Following steps are evaluated for $10^6$ randomly generated states: (i)
 State parameters are chosen randomly in the allowed range. There are three state parameters for the pure w-class state (\ref{wclass}) (taking normalization into account) and 5 state parameters for the pure GHZ-class (\ref{GHZclass}). (ii) We then numerically maximize the FGI (\ref{CP2DI}) over the measurement parameters of the untrusted parties (Alice and Bob). Charlie's measurements are as usual $\sigma_x$ and $\sigma_y$.
FGI is maximally violated numerically if the violation is more than or equal to 3.99. Note that if we keep the same precision for maximum violation i.e. FGI is maximally violated if the violation is above 3.99999 then that are stricter conditions and are already a part of our observations with aforementioned relaxed conditions.\\
\textit{Examining equality of measurement parameters:} All the states that maximally violates the FGI, their state parameters and the measurement parameters are printed out in distinct row in a file. Using \textbf{sort filename | uniq -c} command that outputs the measurement parameters (for one state, eight measurement parameters are in one line and task is to examine whether any two or more lines in the files are same) in ascending order along with their repeated values. If the row is not repeated, it outputs 1 otherwise the number of time it is repeated.  \\
\begin{itemize}
\item \textit{Pure W-class state:} A general pure w-class state has the following form:
\begin{equation}
	\ket{\psi_w} = \sqrt{a}\ket{001} + \sqrt{b}\ket{010} + \sqrt{c}\ket{100} +\sqrt{d}\ket{000}
\label{wclass}
\end{equation}
We observed that out of $10^6$ randomly generated states, 44001 states maximally violates the FGI but no two states have the same set of measurement parameters for untrusted side.
\item \textit{Pure GHZ-class state:}
A general pure GHZ-class state has the following form:
\begin{equation}
\ket{\psi_{GHZ}^{C}}: \cos \delta \ket{000} + \sin \delta e^{\iota \phi} \ket{\phi_A\phi_B\phi_C}
\label{GHZclass}
\end{equation}
where, $\ket{\phi_A} = \cos \alpha \ket{0} + \sin \alpha \ket{1}$, $\ket{\phi_B} = \cos \beta \ket{0} + \sin \beta \ket{1}$ and $\ket{\phi_C} = \cos \gamma \ket{0} + \sin \gamma \ket{1}$. The above state is a GGHZ state (\ref{GGHZ}) for $\delta = \theta$, $\phi = 0$, $\alpha = \beta = \gamma = \frac{\pi}{2}$.\\
We observed that only 6879 states out of $10^6$ maximally violates the FGI but no two states have the same set of measurement parameters for the untrusted sides. \\
\item \textit{Both w-class and GHZ-class pure state:} Even if we take both GHZ-class and w-class pure states together, we found no two states have the same set of measurement parameters for the untrusted sides.

\textit{*The datasets generated during and/or analysed during the current study are available from the corresponding author on reasonable request.}

  \end{itemize}

\end{proof}
\end{widetext}

\bibliography{Tri} 
\end{document}